%% file: ArXiv_ICALP.tex
\documentclass[11pt, oneside,cleveref, autoref]{article}   	
\usepackage{geometry}                		
\usepackage[scr=boondoxo,scrscaled=1.05]{mathalfa}

\usepackage{amssymb}

\input{macros}

\usepackage[bibliography=common]{apxproof}
\usepackage{stmaryrd}
\usepackage{amsmath}
\usepackage{amssymb}
\newtheorem{notation}{Notation}
\newtheorem{definition}{Definition}
\newtheorem{remark}{Remark}
\newtheorem{proposition}{Proposition}
\newtheorem{example}{Example}
\newtheorem{theorem}{Theorem}
\newtheorem{corollary}{Corollary}
\newtheorem{lemma}{Lemma}

\newcommand{\arxiv}[1]{#1}
\newcommand{\noarxiv}[1]{}

\title{Recursion and proof theoretical characterizations of small circuit classes with modulo counting via discrete differential equations} 

\author{Melissa Antonelli \quad Arnaud Durand \quad Rui Li}

\begin{document}
\maketitle

\begin{abstract}
    The paper proposes an implicit (i.e., machine-independent) complexity approach to studying computation by polynomial-size, constant-depth circuits with gates counting modulo a constant through the lens of discrete ordinary differential equations (ODEs). So far, recursion-theoretic characterizations have been provided for functions computed by circuits of constant depth, including gates counting modulo 2 and 6 only (i.e.,~for the classes $\ACCt$ and $\ACCp{6}$, resp.). In this paper, it is shown that considering ODE schemas, rather than bounded recursion, allows for a more fine-grained analysis, leading to (uniform) characterizations for \emph{all} classes $\ACCp{n}$ ($n \in \Nat$), i.e.~functions computed by circuits including counting modulo $n$ gates. 
    Inspired by the syntactic form of the ODE schemas, we go further in this direction and present first-order bounded theories for capturing provably total functions in each of these classes.
\end{abstract}

\section{Introduction}
Circuit complexity has been studied extensively from multiple angles in recent decades.
Just to mention a concrete objective, understanding the computational power of polynomial-size, $\log^k n$ depth circuit classes, $\mathbf{AC}^k$, and their bounded fan-in counterpart, $\mathbf{NC}^k$, has received significant  attention.
Although some of the few known lower bounds have been established precisely in relation to Boolean and arithmetic circuits,
%
%
many \emph{fundamental} questions in this field are still open; for example, the separation between levels of the mentioned hierarchies has currently been established only for levels 0 and 1.
Indeed, from the  celebrated proofs of~\cite{DBLP:conf/focs/FurstSS81,DBLP:conf/stoc/Smolensky87}, it is known that Parity, the set of words having an even number of 1's, cannot be decided by polynomial-size, constant-depth circuits, i.e.,~in $\mathbf{AC^0}$.
This result illustrates the (almost) complete inability of constant-depth circuits to count, motivating the study of their extension via counting modulo $n \in \mathbb{N}$ or majority gates, which gave rise to the classes $\mathbf{AC}^0[n]$ and $\mathbf{TC}^0$, resp.
In this setting, the parity lower bound directly implies that  $\mathbf{AC}^0 \subsetneq \mathbf{AC}^0\mathbf{[2]}$.  
More recently, related questions have gained interest even in less traditional fields of computer science; for instance, driven by the link between circuit resources and neural networks~\cite{Parberry},  as early as the 1990s, the expressive power of feed-forward and graph networks was related to threshold Boolean circuits~\cite{AllenderWagner,MSS} and constant-depth arithmetic circuits over the reals~\cite{BHSVV}, resp.

An alternative approach to studying these foundational questions is the investigation of circuit expressive power from machine-independent perspectives; namely, through the prism of logic (especially finite model theory~\cite{DHV,HKL}), proof complexity, and recursion theory.
The contributions of this paper fall precisely within this line of research, starting from a newly proposed perspective in the latter area~\cite{BournezDurand}.
Generally speaking, in recursion theory capturing complexity classes has been done by considering (closure under) restricted forms of primitive recursion.
This is, for example, the idea behind Cobham's characterization of poly-time computable functions ($\FP$) via the so-called \emph{bounded recursion on notation} (BRN) schemas~\cite{Cobham}.
Following this approach, in a series of works~\cite{Clote88, Clote1990,Clote93}, Clote introduced the algebra $\mathcal{A}_0$ to capture functions in the log-time hierarchy, equivalent to $\AC^0$ (the function analogue of $\mathbf{AC}^0$), 
using as a main tool a restricted form of (bounded) recursion, called \emph{concatenation recursion on notation} (CRN). 
This result was extended to higher levels in the hierarchies $\AC$ and $\NC$ (the function versions of $\mathbf{AC}$ and $\mathbf{NC}$).
Around the same time, an alternative algebra (and logic), based on a schema called \emph{upward tree recursion}, was introduced in~\cite{ComptonLaflamme} to characterize $\NC^1$, while in~\cite{Allen}, an original recursion-theoretic characterization (and arithmetic \emph{\`a la Buss}) was presented to capture $\NC$.
Other small circuit classes, including $\TC^0$, the class of functions computed by polynomial-size, constant-depth unbounded fan-in circuits with threshold gates, were later considered by Clote and Takeuti, who also introduced first-order bounded theories for some of them~\cite{CloteTakeuti}. 
This approach allowed for the characterization of the classes $\ACCt$ and $\ACCp{6}$, but it failed to extend to other modulo classes.

In the broader context of implicit characterizations of complexity, even beyond circuit-based models, a new, machine-independent approach via differential equations has recently emerged and has been shown to be suitable for both discrete~\cite{BournezDurand} and continuous settings~\cite{BlancBournez22,BlancBournez23,BlancBournez24}.
In the discrete case, one can describe a function by an ordinary differential equation (ODE) it has to satisfy.
Informally, given three functions $g,h$ and $t$, a function $f$  is said to be the solution of a (generalized) ODE (a.k.a.,~initial value problem) if it satisfies
$$
\frac{\partial f(x,\tu y)}{\partial t(x)} = h\big(x, \tu y, f(x, \tu y)\big)
$$
together with some initial condition $f(0, \tu y)=g(\tu y)$.
Recall that here $\frac{\partial f(x, \tu y)}{\partial t(x)}$ stands for the derivative of $f(x, \tu y)$ along the values of $t(x)$ (the simpler case of the \emph{discrete} derivative of $f(x, \tu y)$ on the variable $x$ is simply $\frac{\partial f(x, \tu y)}{\partial x}= f(x+1,\tu y) - f(x, \tu y)$).
Intuitively, this new framework  to define functions offers \emph{three different levels of control}:
\begin{itemize}
    \item the choice of $t$, the function to derive along, influences the \emph{growth speed}, i.e.~the number of steps, or (in parallel context)  the \emph{width} of computation;   
    \item the function $h$ (and $g$) controls the \emph{growth} of intermediate values and the \emph{feasibility of iterating} elementary computation steps;
    \item constraints on calls to $f$ in $h$, e.g.~allowing ``unrestricted'' calls to $f$ or variously limiting them, influence \emph{memory} use or \emph{parallelization depth}. 
\end{itemize}
This framework appears to be significantly more flexible and expressive than standard limited recursion, imposing bounds on the value $h$. Additionally, unlike the latter, 
 the ODE approach offers a way to talk about discrete and continuous computations in a relatively similar way, often focusing on similar syntactical restrictions~\cite{BlancBournez23,BlancBournez24}.

Such an approach was initially used to establish a correspondence between functions computable in poly-time and solutions of specific ODEs,
whose defining equations are of a special (essentially) linear form and obtained by deriving along functions of specific growth rate~\cite{BournezDurand23} (Sec.~\ref{sec:ODEs}). 
In this vein, previous investigation on small circuit classes has already demonstrated a, quite striking, correspondence between simple syntactical constraints on these ODE schemas and computation performed by circuits, allowing for a fine-grained definition of the corresponding classes.
However, although different circuit classes have been considered in~\cite{ADK24a,ADK25}, the study of those that include counting gates has been developed only partially and has been so far restricted to $\TC^0$ and $\ACCt$.

\paragraph{Contributions of the paper.}
In this paper, we push this investigation further by characterizing classes allowing various forms of counting, which have never been captured via algebras (or first-order bounded arithmetic) before.
Concretely we introduce \emph{natural} and \emph{uniform} ODE-based function algebras,
for $\ACCp{n}$ (being $n\in \Nat$), the class of functions computable by polynomial-size and constant-depth circuits with counting modulo $n$ gates.
The key ingredient is an ODE-based recursion schema, whose evaluation essentially requires modulo counting.
Despite its apparent simplicity, this change in perspective allows for a significant advance over the existing literature since, to the best of our knowledge, so far implicit (recursion-theoretic) characterizations have been provided only for specific classes, namely $\ACCt$ and $\ACCp{6}$~\cite{CloteTakeuti}.
Additionally, inspired by these schemas, we introduce new bounded arithmetic theories to syntactically capture these same classes.
Again, this goes beyond what is currently known (except for $n=2$ and $n=6$).
More generally, the investigation presented here aims to clarify the connection between features of ODE schemas and resources in circuit computation. 
This mirrors insights obtained in descriptive complexity, which provides a correspondence between resources in circuits or parallel machines and expressive power of logical features (like the number of variables)~\cite{Immerman,Immerman89}, but extends to broader forms of computation, for example including (modulo) counting, which have never been fully addressed in other frameworks.
As an overall goal, this work aims to provide an original perspective for investigating core questions in (circuit) complexity, such as class separations and lower bounds.

\paragraph{Structure of the paper.}
Our presentation will proceed as follows.
In Sec.~\ref{sec:preliminaries}, we introduce notational conventions, basic notions and methodologies needed for the rest of the paper.
In Sec.~\ref{sec:algebras}, we present new, ODE-based function algebras capturing poly-size and constant-depth circuit classes, including modulo 2 (Sec.~\ref{sec:ACCt}), more general modulo (arbitrary) $n$ (Sec.~\ref{sec:ACCn}), and majority gates (Sec.~\ref{sec:TC}).
A formal comparison with the few existing characterizations is also presented. 
Moving to proof theory, in Sec.~\ref{sec:bArit}, we show the technical advantages of relying on these function algebras to deal with a known first-order bounded arithmetic for $\AC^0$ (Sec.~\ref{sec:TAC}) and present alternative original theories inspired by the ODE-based schemas to capture provably total functions in $\ACCp{n}$ (Sec.~\ref{sec:BA} and~\ref{sec:BA characterization with new rule}).
We conclude by pointing to open problems and future directions of research (Sec.~\ref{sec:conclusion}).
Due to space constraints, most proofs are omitted but some are available in the Appendix.

\section{Preliminaries}\label{sec:preliminaries}

\arxiv{
In this Section, we \arxiv{briefly} introduce the preliminary notion needed to present our new results. 
In Sec.~\ref{sec:ODEs}, we recall the basics of the ODE-based approach to complexity.
In Sec.~\ref{sec:circuit}, we recall the definition of the complexity classes we are going to characterize and the state of the art regarding recursion-theoretic characterizations known for them.
Finally, in Sec.~\ref{sec:circODE}, we summarize the main characterizations for small circuit classes presented so far.}

\arxiv{
\begin{notation}
    Throughout this paper, we will adopt the following notational conventions.
    We use $s(x)=x+1$ to denote the successor, while $s_0(x)=2x$ and $s_1(x)=2x+1$ are the binary successor functions.
    As standard, $x\#y=2^{\ell(y)}\times x+y$ is the smash function, $\div n$ represents integer division by $n\in \N^*$, i.e.~for all $x\in \mathbb{Z}$, $x\div n =\big\lfloor \frac{x}{n} \big\rfloor$ and, for $i,k\in\N$, $\pi_i^k(x_1,...,x_k)=x_i$ is the projection function.
    For $x\in \Nat^+$, $\ell(x)$ is the length function which, given an input $x$, returns its length written in binary, i.e.~$\lceil$log$_2(x+1)\rceil$ and such that (for convenience) $\ell(0)=0$; $\fun{BIT}$ is the bit function so defined that BIT$(i,x)=\big\lfloor \frac{x}{2^i}\big\rfloor$ \textsc{mod 2} returns the value of the $i^{th}$ bit in the binary representation of $x$.
    The sign function (over $\mathbb{Z}$), $\fun{sg} : \mathbb{Z} \to \mathbb{Z}$ is so defined that $\fun{sg}(x)=1$ if $x>0$ and is equal to $0$ otherwise, while $\fun{cosg}(x)$ is equal to $1-\fun{sg}(x)$.
    Finally, for $x \in \Nat^{+}$, we will use $\alpha(x)=2^x-1$ to denote the greatest integer whose length is $x$.
\end{notation}
}

\subsection{A gentle introduction to discrete ODEs (in complexity)}\label{sec:ODEs}

Difference calculus, the discrete analogue of differential calculus, is an old field of mathematics~\cite{gelfand1963calcul,jordan1965calculus}.
Recall that the \emph{discrete derivative of $f(x)$} is defined as $\Delta f(x)= f(x+1) - f(x)$ and that ODEs are expressions of the form:
$ 
\frac{\partial f(x,\tu y)}{\partial x} = h\big(x, \tu y, f(x, \tu y)\big)
$ 
where $\frac{\partial f(x, \tu y)}{\partial x}$ stands for the derivative of $f(x, \tu y)$ considered as a function of $x$ for a fixed value of $\tu y$.
When some initial value $f(0, \tu y)=g(\tu y)$ is added, this form an Initial Value Problem (IVP).

As anticipated, this framework is shown to offer natural ways to capture primitive recursion and function classes. 
Specifically, in order to deal with complexity, this picture must be endowed with new notions, starting with that of derivation \emph{along a function}:

\begin{definition}[$\lambda$-ODE]
Let $f, \lambda:\Nat^p \to \mathbb{Z}$ and $h:\Nat^{p+1} \to \mathbb{Z}$ be functions. Then,
\begin{align}
\frac{\partial f(x,\tu y)}{\partial \lambda} = \frac{\partial f(x, \tu y)}{\partial \lambda (x, \tu y)} = h\big(x, \tu y, f(x, \tu y)\big)
\end{align}
%
is a formal synonym of
%
%
\begin{align*}
f(x + 1, \tu y) &= f(x, \tu y) + \big(\lambda (x+1, \tu y) - \lambda (x, \tu y)\big) \times h\big(x, \tu y, f(x, \tu y)\big) 
\\
 &= f(x, \tu y) +  \frac{\partial \lambda(x,\tu y)}{\partial x} \times h\big(x, \tu y, f(x, \tu y)\big).
\end{align*}
When $\lambda (x,\tu y)=\ell(x)$, we call (1) \emph{length-ODE} or $\ell$-ODE.
\end{definition}

\noindent
Intuitively, one of the key properties of the $\lambda$-ODE schema is its dependence on the number of distinct values of $\lambda$, i.e.,~the value of $f\arxiv{(x,y)}$ changes only when the value of $\lambda\arxiv{(x,y)}$ does.
In this way, the growth rate of $\lambda$ determines the number of steps needed to compute the value of $f$.
Computation of solutions of $\lambda$-ODEs has been fully described in~\cite{BournezDurand23}.
Here, we focus on the special case of $\lambda=\ell$, which is particularly interesting for our characterizations since the value of $\ell(x)$ changes, increasing by 1, only when $x$ goes from $2^t-1$ to $2^t$ (for $t\ge 1$).
Setting $h\big(\alpha(-1), \tu y, \cdot\big)= f(0, \tu y)$, it is shown by induction that
$
f(x, \tu y) = \sum^{\ell(x)-1}_{u=-1} h\big(\alpha(u), \tu y, f(\alpha(u), \tu y)\big)
$, see App.~\ref{app:ODE}.

\begin{toappendix}
\subsection{Details of Section~\ref{sec:ODEs}}\label{app:ODE}
If $f$ is the solution of (1) with $\lambda=\ell$ and has initial value $f(0, \tu y)=g(\tu y)$, then, in particular, $f(1, \tu y)=f(0, \tu y)+ \big(\ell(1)-\ell(0)\big) \times h(0, \tu y, f(0, \tu y))=f(0, \tu y)+ h(0, \tu y, f(0, \tu y))$. More generally, 
$f(x,\tu y)=f(\alpha(\ell(x)),\tu y)$ since $x$ and $\alpha(\ell(x))$ have the same length and then $f(\alpha(\ell(x)),\tu y)=
f(\alpha(\ell(x)-1), \tu y) + \big(\ell(x) - \ell(x-1)\big) \times h\big(\alpha(\ell(x)-1), \tu y, f(\alpha(\ell(x)-1), \tu y)\big)
= 
f(\alpha(\ell(x)-1), \tu y) + h\big(\alpha(\ell(x)-1), \tu y, f(\alpha(\ell(x)-1), \tu y)\big)$.
Starting from $t=x\ge 1$ and taking successive decreasing values of $t$, the first non-zero difference $\ell(t)-\ell(t-1) (\neq 0)$ is given by the biggest $t-1$ such that $\ell(t-1)=\ell(t)-1$, i.e.,~$t-1=\alpha(\ell(t)-1)$.
Setting $h\big(\alpha(-1), \tu y, \cdot\big)= f(0, \tu y)$, it follows by induction that:
$$
f(x, \tu y) = \sum^{\ell(x)-1}_{u=-1} h\big(\alpha(u), \tu y, f(\alpha(u), \tu y)\big).
$$
\end{toappendix}

Obviously, the choice of the function $h$ also influences the complexity of the computation.  
Here, we will mostly deal with (limited) $\fun{sg}$-polynomial expressions, where a \emph{(resp. limited) $\fun{sg}$-polynomial expression} is an expression over the signature $\{+,-,\div 2, \times\}$ (resp. $\{+,-,\div 2\}$) on the set of variables/terms $X= \{x_1, \dots, x_h\}$ and integer constants (see~\cite{ADK24a}). 
The degree of a $\fun{sg}$-polynomial expression~$\deg(P)$ is inductively defined so that $\deg(P)=0$ for $P$ being a constant, a variable or a signed expression of the form $\fun{sg}(Q)$, 
$\deg(\tu x, Q\div 2)=\deg(\tu x, Q)$,  $\deg(\tu x, Q*R)$ =
		$\max\{\deg(\tu x, Q),$ $\deg(\tu x, R)\}$, for $*\in \{+,-\}$, and
		 $\deg(\tu x, Q\times R)= \deg(\tu x, Q) + \deg(\tu x, R)$.
%
%
%
%
%

\begin{toappendix}
\begin{definition}[(Limited) polynomial expression, $\fun{sg}$(-limited) polynomial expression]
	A \emph{($\fun{sg}$-)limited polynomial expression} is an expression over the signature $\{+, -, \div 2\}$ (and the $\fun{sg}$ function, resp.) on the set of variables/terms $X= \{x_1, \dots, x_h\}$ and integer constants; a \emph{($\fun{sg}$-)polynomial expression} is defined analogously over the signature $\{+, -, \div 2,\times\}$.
\end{definition}

\noindent
To define the notion of degree of polynomial expressions, we follow the generalization to sets of variables introduced in~\cite{BournezDurand}.

\begin{definition}
	Given a list of variables $\tu x=x_{i_1}, \dots, x_{i_m}$, for $i_1, \dots i_m\in \{1, \dots, h\}$, the \emph{degree of a set $\tu x$ in a $\fun{sg}$-polynomial expression $P$}, $\deg(\tu x, P)$, is inductively defined as follows:
	\begin{itemize}
		\itemsep0em
		\item $\deg(\tu x, x_{i_j})=1$, for $x_{i_j}\in \{x_{i_1}, \dots, x_{i_m}\}$, and $\deg(\tu x, x_{i_j})=0$, for $x_{i_j}\not\in \{x_{i_j}, \dots, x_{i_m}\}$
		\item $\deg(\tu x, \fun{sg}(Q))=0$
		\item $\deg(\tu x, Q\div 2)=\deg(\tu x, Q)$
		\item $\deg(\tu x, Q+R)$ = $\deg(\tu x, Q-R)$ =
		$\max\{\deg(\tu x, Q),$ $\deg(\tu x, R)\}$
		\item $\deg(\tu x, Q\times R)= \deg(\tu x, Q) + \deg(\tu x, R)$.
	\end{itemize}
\end{definition}
\end{toappendix}
\noindent
A $\fun{sg}$-polynomial expression $P$ is said to be \emph{essentially constant in a set of variables $\tu x$} when $\deg(\tu x, P)=0$ and \emph{essentially linear in $\tu x$}, when $\deg(\tu x, P)=1$.
The concept of \emph{linearity} allows us to control the growth of functions defined by ODEs.

\begin{definition}[Linear $\lambda$-ODE]\label{def:ODE}
	Given $g:\Nat^p \to \Nat$, $h, \lambda:\Nat^{p+1}\to \mathbb{Z}$ 
    and a $\fun{sg}$-polynomial expression $P$, the function $f:\Nat^{p+1} \to \mathbb{Z}$ is \emph{linear $\lambda$-ODE} definable from $g,h$ and $P$ if it is the solution of the IVP with initial value $f(0, \tu y)=g(\tu y)$ and s.t.~$ 
		\frac{\partial f(x, \tu y)}{\partial \lambda} = P\big(x, \tu y, f(x, \tu y), h(x, \tu y)\big)
	$ 
	where $P$ is essentially linear in 
    $f(x, \tu y)$.
	For $\lambda=\ell$, such schema is called \emph{linear length ODE}.
\end{definition}

\noindent
Stated otherwise, if $P$ is \emph{essentially linear in $f(x,\tu y)$}, there exist $A$ and $B$, which are essentially constants in $f(x,\tu y)$, s.t.
\begin{align*}
f(0,\tu y) &= g(\tu y) \\
\frac{\partial f(x, \tu y)}{\partial \ell} &= A(x, \tu y, h(x, \tu y), f(x, \tu y)) \times f(x, \tu y)  + B(x, \tu y, h(x, \tu y), f(x, \tu y))
\end{align*}
Then, for all $x$ and $\tu y$, $f(x, \tu y)$ =
$\sum^{\ell(x)-1}_{u=-1} \Big(\prod^{\ell(x)-1}_{t=u+1} \Big(1+A\Big(\alpha(t), \tu y, h(\alpha(t), \tu y), f(\alpha(t), \tu y)\Big)\Big)\Big)\times B\big(\alpha(u), \tu y, h(\alpha(u), \tu y), f(\alpha(u), \tu y)\big)$,
with the convention that $\prod^{x-1}_x\kappa(x)=1$ and $B\big(\alpha(-1), \tu y, \cdot, \cdot \big) = f(0, \tu y)$.

One of the main results of~\cite{BournezDurand} is the implicit characterization of $\FP$ by the algebra made of basic functions $\fun{0}, \fun{1}, \pi^p_i, \ell, +, -, \times, \fun{sg}$ and closed under composition $(\circ)$ and $\ell$-ODE.

\arxiv{
\subsection{Implicit approaches to circuit complexity}\label{sec:circuit}

This Section summarizes the definitions of the relevant classes and the characterizations previously established in the literature.}
%
%
%

\begin{toappendix}
\subsection{Small circuit classes (with modulo counting gates)}\label{app:circuits}
Boolean circuits are vertex-labelled directed acyclic graphs whose nodes are either input nodes, output nodes or are labelled with a Boolean function in $\{\neg,\wedge,\vee\}$.
A Boolean circuit with modulo gates allows in addition gates labelled as \textsc{Mod 2} or, more generally, as \textsc{Mod $p$}, that output the sum of the inputs modulo 2 or $p$, resp.\footnote{In the context of this paper, it is equivalent to allowing gates returning 1 when the sum of the inputs modulo $p$ is 0 and 0 otherwise.}
Similarly, a threshold circuit includes \textsc{Maj} gates, that output 1 when the majority of their inputs are 1's.
A family of circuits is said to be $\cc{Dlogtime}$-uniform if there is a Turing machine (with a random access tape) that decides in deterministic logarithmic time the direct connection language of the circuit, i.e.,~which, given $1^n, a, b$ and $t\in \{\neg, \wedge, \vee\}$, decides if $a$ is of type $t$ and $b$ is a predecessor of $a$ in the circuit (and analogously for input/output nodes).

When dealing with circuits, the resources of interests are \emph{size}, i.e.,~the number of gates, and \emph{depth}, i.e.,~the length of the longest path from one input to the output (see~\cite{Vollmer} for more details and related results).

\begin{definition}[Classes $\cc{FAC}^i$ and $\cc{FNC}^i$]
    For $i \in \Nat$, $\cc{AC}^i$ (resp.,~$\cc{NC}^i$) is the class of languages recognized by a $\cc{Dlogtime}$-uniform family of Boolean circuits of unbounded (resp., bounded) fan-in gates of polynomial size and depth $O(($log $n)^i)$.
    We denote by $\cc{FAC}^i$ (resp., $\cc{FNC}^i)$ the corresponding function class.
\end{definition}

\noindent
It is known that unbounded (poly($n$)) fan-in can be simulated by bounded trees of depth $O$(log $n$), so that $\cc{FNC}^i\subseteq \cc{FAC}^i \subseteq \cc{FNC}^{i+1}$.
The inclusion has been proved to be strict only for $i=0$.
In contrast, when dealing with classes with counting abilities, interesting bounds have been established.

\begin{definition}[Classes $\ACCp{n}$ and $\TC$]
\begin{sloppypar}
    The class $\ACCp{n}$ (resp., $\TC$) is the class of languages recognized by a $\cc{Dlogtime}$-uniform family of Boolean circuits of polynomial size and constant depth including \textsc{Mod $p$} (resp., \textsc{Maj}) gates.
    We denote by $\ACCp{n}$ (resp., $\TC^0$) the corresponding function classes.
    Additionally $\cc{FACC^0}=\bigcup^\infty_{m=2} \ACCp{m}$.
\end{sloppypar}
\end{definition}

\noindent
As mentioned, it is a foundational result that the inclusion $\AC^0\subset \ACCp{2}$ is proper, as the parity function cannot be computed in constant depth by standard unbounded fan-in Boolean gates of polynomial size, but it can be if \textsc{Mod 2} gates are added.
More generally, it is known that for any distinct prime numbers $p,q\in \Nat$, $\ACCp{p} \neq \ACCp{q}$, and if $p'=p\times q$ then $\ACCp{p} \subset \ACCp{p'}$.
Furthermore, since any \textsc{Mod $p$} gate can be simulated by \textsc{Maj} in constant depth, $\ACCp{p} \subset \cc{FACC^0} \subseteq \TC^0$ (see~\cite{DBLP:conf/stoc/Smolensky87}).
\end{toappendix}

\subsection{Function algebras for small circuit classes}\label{sec:circODE}

\arxiv{\subsection{Function algebras for small circuit classes}}


\noindent
\textbf{Circuit complexity through the lens of bounded recursion.}
We assume that the reader is familiar with the standard notions of small circuit classes (see~\cite{Vollmer}).
Here we briefly recall main recursion theoretic characterizations offered in the literature.
To the best of our knowledge, when considering small circuit classes including gates counting modulo $n\in \Nat$, only a few implicit characterizations have been provided so far: function algebras and first-order bounded arithmetics by Clote and Takeuti~\cite{CloteTakeuti}, a second-order family of theories by Cook and Nguyen~\cite{CookNguyen} and Johannsen's theory for 
$\ACCt$. 
Since we will often compare our results with those in~\cite{Clote1990,CloteTakeuti}, we briefly recall them.
%

\begin{definition}[CRN, $k$-BRN] 
Let $i\in\{0,1\}$, $g:\Nat^p\to \Nat$ and $h_i:\Nat^{p+2}\to \Nat$. A function $f:\Nat^{p+1}\to \Nat$ is said to be defined from $g,h_0$ and $h_1$ by
\begin{itemize}
    \item \emph{Concatenation Recursion on Notation} \emph{(CRN)} if for all $x$ and $\tu y$:
    $f(0,\tu y) = g(\tu y)$ and $f(s_i(x), \tu y) = s_{h_i(x,\tu y)}(f(x,\tu y))$, where $h_i$ takes values in $\{0,1\}$
    \item \emph{$k$-Bounded Recursion on Notation} $(k$\emph{-BRN)} if for all $x$ and $\tu y$: $f(0, \tu y) = g(\tu y)$ and $f(s_i(x), \tu y) = h_i(x,\tu y, f(x, \tu y))$, with $f(x,\tu y) \leq k$ for a constant $k$.
\end{itemize}
\end{definition}


In~\cite{Clote1990}, it is proved that  the function algebra $\mathcal{A}_0=[0, \pi^p_i, s_0, s_1, \ell, \text{BIT}, \#; \circ, \text{CRN}]$ captures $\AC^0$. Enriching $\mathcal{A}_0$ with 1-BRN gives a characterization for $\ACCt$, with 2-BRN captures $\ACCp{6}$ and for $k$-BRN, with $k\ge 4$, characterizes $\NC^1$~\cite[Th.~2.2, 2.5]{CloteTakeuti}. 
Stated otherwise, by restricting the range of the schema characterizing $\NC^1$, i.e.~$k$-BRN, two specific counting classes are captured but no generalization to $\ACCp{n}$ for $n\not\in \{2,6\}$ can be expected.
%
%
A function algebra for $\TC^0$ is there defined by simply endowing $\mathcal{A}_0$ with multiplication.
%
Also,
first-order bounded theories are introduced to syntactically characterize all these classes (see Sec.~\ref{sec:bArit}).


\arxiv{
In~\cite[Th. 2.2, 2.5]{CloteTakeuti} the following characterizations are provided:
\begin{align*}
    \ACCt &= [0, \pi^p_i, s_0, s_1, \ell, \text{BIT}, \#; \circ, \text{CRN}, \text{1-BRN}] \\
    \ACCp{6} &= [0, \pi^p_i, s_0, s_1, \ell, \text{BIT}, \#; \circ, \text{CRN}, \text{2-BRN}] \\
    \NC^1 &= [0, \pi^p_i, s_0, s_1, \ell, \text{BIT}, \#; \circ, \text{CRN}, k\text{-BRN}].
\end{align*}
Notably, $\ACCt$ and $\ACCp{6}$ are defined by simply restricting the range of $k$ in $k$-BRN, the schema characterizing $\NC^1$.}


\bigskip
\noindent
\textbf{Circuit complexity through the lens of ODEs.}
We present an overview of the characterizations of small circuit classes (without modulo counting gates) obtained relying on\arxiv{different, strict and non-strict,} ($\ell$-)ODE schemas.
Although different classes have been considered in~\cite{ADK24a,ADK25}, investigation of those that include counting gates has been developed only partially. 
In particular, for counting modulo 2, an algebra was  defined in terms of non-strict schemas, which, as we shall see, captures 
constant-depth computation somewhat indirectly.

\begin{notation}[Strict and simple schemas]
    An equation or, by extension, an ODE schema is said to be \emph{strict} when it does not include any call to $f$ \arxiv{$f(x,\tu y)$ }in $A$ and $B$ (not even under the scope of the sign function).
    %
    %
    For the sake of clarity, we will use $k(x,\tu y)$ (resp., $K(x, \tu y, h(x,\tu y), f(x,\tu y)))$ for functions (resp. expressions) with restricted range, typically in $\{0,1\}$.
    \arxiv{For the sake of clarity, we will use $k(x,\tu y)$ for functions with restricted range, typically in $\{0,1\}$, and, analogously, $K(x,\tu y, h(x, \tu y), f(x, \tu y))$ for (possibly limited) $\fun{sg}$-polynomial expressions of restricted range.}
\end{notation}

\noindent
If not stated otherwise, 
\arxiv{we assume that} 
the expressions used in schemas characterizing $\AC^0$ and $\ACCp{n}$ are \emph{limited} $\fun{sg}$-polynomial expressions (multiplication is not computable in these classes).

\begin{notation}
    Let $\mathscr{B}_0$ denote the set of basic functions $\fun{0}, \fun{1}, \fun{sg}, \ell, +, -, \div 2, \#, \fun{BIT}, \pi^p_i$.
\end{notation}

\arxiv{
\begin{notation}
    In what follows we will use $\mathscr{B}_0$ to denote the set of basic functions made of $\fun{0}, \fun{1}, \fun{sg}, \ell, +, -, \div 2, \#, \fun{BIT}, \pi^p_i$.
\end{notation}
}

The function algebra capturing $\AC^0$ is obtained by 
the strict 
schema below~\cite{ADK24a,ADK25}.

\begin{definition}[$\ell$-ODE$_1$ schema]
Given $g:\Nat^p \to \Nat$ and $k:\Nat^{p+1} \to \{0,1\}$, the function $f:\Nat^{p+1} \to \Nat$ is obtained by $\ell$-\emph{ODE}$_1$ from $g$ and $k$, if it is the solution of the IVP with initial value $f(0,\tu y)=g(\tu y)$ and s.t.~$ 
\frac{\partial f(x,\tu y)}{\partial \ell} = f(x,\tu y) + k(x,\tu y).
$
\end{definition}

\noindent
Computation through $\ell$-ODE$_1$ intuitively corresponds to iteratively left-shifting (the binary representation) of a  number, each time possibly adding 1 in the last position.
%
Due to it, the algebra $\ACDL = [\mathscr{B}_0; \circ, \ell$-ODE$_1]$
is introduced and proved to correspond \arxiv{precisely}to \arxiv{the class}$\AC^0$.\arxiv{\footnote{Actually, in~\cite{ADK24a}, $\ACDL$ is defined relying on a second ODE schema and without taking $\fun{BIT}$ as a basic function of the algebra; in~\cite{ADK25} this definition and the one above are proved equivalent.}.}

\begin{remark}[Bounded quantification and minimization]\label{remark:minimisation}
Sharply bounded quantification and minimization can be naturally rewritten in $\ACDL$.
%
Let $R\subseteq \Nat^{p+1}$ and $h_R$ be its characteristic function.
    For all $x$ and $\tu y$, it holds that $(\exists z\leq \ell(x)) R(z, \tu y)=\fun{sg}(f_\exists(x,\tu y))$, where $f_\exists$ is defined as the solution of the  IVP with initial value $f_\exists(0, \tu y)=h_{R_0}(\tu y)$ (being $h_{R_0}(\tu y) = 0$ if $R(0,\tu y)=0$ and $h_{R_0}(\tu y) =1$ if $R(0,\tu y)=1$) and such that:
    $$
        \frac{\partial f_\exists(x, \tu y)}{\partial \ell}= f_\exists(x, \tu y) + h_{R}(\ell(x)+1, \tu y).
    $$
    Clearly, this is an instance of $\ell$-ODE$_1$.
    Intuitively, $f_\exists (x,\tu y)\neq 0$ when, for some $z$ smaller than $\ell(x)$, $R(z,\tu y)$ is satisfied (i.e.,~$h_R(z,\tu y)=1$): if such instance exists, our bounded search ends with a positive answer. 
    
%
    In order to express universally sharply bounded quantification in $\ACDL$ we consider $(\forall z\leq \ell(x))R(z,\tu y)=\fun{cosg}(f_\forall(x,\tu y))$ for $f$ such that:
    \begin{align*}
        f(0, \tu y) &= \fun{cosg}(h_{R_0}(\tu y)) \\
        \frac{\partial f_\forall(x,\tu y)}{\partial \ell} &= f_\forall(x,\tu y) + \fun{cosg}(h_R(\ell(x)+1), \tu y)
    \end{align*}

    Finally, the sharply bounded $\mu$-operator, $\mu x.R(x,\tu y)$, returning the smallest $z\leq \ell(x)$ such that $R(z,\tu y)$ holds, can be expressed in $\ACDL$ by considering $\fun{mu}_R(x,\tu y) = \ell(x) - \ell(f_\mu(x,\tu y))$, where $f_\mu$ is the solution of the IVP with initial value $f_\mu(0,\tu y)=0$ and such that:
    $$
    \frac{\partial f_\mu(x,\tu y)}{\partial \ell} = f_\mu(x,\tu y) + \exists i \leq \ell(x+1) \big(h_R(i, \tu y) = 1).
    $$
    Since, as seen, existentially sharply bounded search can be expressed in $\ACDL$, this is an instance of $\ell$-ODE$_1$.
\end{remark}

\begin{toappendix}
The following schema generalizes $\ell$-ODE$_1$ by allowing for a more liberal left-shifting operation which, in turn, enables for encoding of sequences of numbers (of bounded length), while remaining in $\ACDL$.

\begin{definition}[$\ODEup$ schema] \label{df:upODE}
    Let $g:\Nat^p \to \Nat, k:\Nat^{p+1} \to \{0,1\}$ and $h:\Nat^p \to \Nat$, such that $h(\tu y) \ge \max(1,k(x,\tu y))$. The function $f: \Nat^{p+1}\to \Nat$ is defined by $\ODEup$ from them if it is the solution of the IVP s.t.~$f(0, \tu y)=g(\tu y)$ and $\frac{\partial f(x,\tu y)}{\partial \ell} =\big(2^{\ell(h(\tu y))}-1\big) \times f(x, \tu y) + h'(x,\tu y).$
    \arxiv{such that:
    $$ 
    \frac{\partial f(x,\tu y)}{\partial \ell} =\big(2^{\ell(h(\tu y))}-1\big) \times f(x, \tu y) + h'(x,\tu y).
    $$}
\end{definition}

%


\begin{proposition}\label{prop:ODEup}
    If $f$ is defined by $\ODEup$ from functions \arxiv{computable}in $\AC^0$, then $f$ is\arxiv{computable} in $\AC^0$.
\end{proposition}

\begin{proof}[Proof of Prop.~\ref{prop:ODEu}]
    By definition of $\ell$-ODE (Df.~\ref{def:ODE}), the solution of the corresponding system can be computed by a constant-depth circuit first calculating $g(\tu y)$ and, for any $t\in \{0, \dots, x\}$, $h'(t,\tu y)$'s (in constant depth for hypothesis), then left-shifting these results of respectively $2^{\ell(h(\tu y))}$ bits and, finally, concatenating them (as no carry occurs).
\end{proof}

\begin{proof}
Recall that, by definition of $\ODEup$, $h(\tu y)\ge 1$.
    By Definition of $\ell$-ODE, for every $x$ and $\tu y$, the solution of this system is of the form:
    \begin{align*}
        f(x,\tu y) &= \sum^{\ell(x)-1}_{u=-1} \bigg(\prod^{\ell(x)-1}_{t=u+1} 2^{\ell(h(\tu y))}\bigg) \times h'(\alpha(u),\tu y) \\
        &= \sum^{\ell(x)-1}_{u=-1} 2^{\ell(h(\tu y)) \times (\ell(x)-1+u)} \times h'(\alpha (u), \tu y)
    \end{align*}
    with the convention that $\prod^{x-1}_x\kappa(x)=1$ and $h'(\alpha(-1), \tu y)=f(0, \tu y)$.
    The multiplication involved here is only by a power of 2, which basically corresponds to left-shifting (which is computable in $\AC^0$).
    As desired, the function can be computed by the following unbounded fan-in Boolean circuit in constant depth:
    \begin{itemize}
    \begin{sloppypar}
        \item In parallel, compute the values of $g(\tu y)$, for each $t\in \{0,\dots, \ell(x)-1\}$, $\ell(h(\tu y)\times (\ell(x)-1-t)$ and $h'(\alpha(t), \tu y)$.
        This can be done in constant depth by hypothesis.
        \end{sloppypar}
        \item In one step, left-shift the value of $g(\tu y)$ by padding $\ell(h(\tu y))\times \ell(x)-1$ 0's on the right and, for $t\ge 0$, shift all values $h'(\alpha(t), \tu y)$ by padding $\ell(h(\tu y)) \times (\ell(x)-1-t)$ 0's  on the right  (i.e.,~by multiplying by 2$^{\ell(h(\tu y)\times \ell(x)-1-t}$) and by padding $\ell(g(\tu y))+\ell(k(\tu y))\times (u+1)-1$ 0's on the left.
        Intuitively, the fact that by definition $\ell(h(\tu y)) \ge \ell(h'(x,\tu y))$ ensures that each value $h'(x,\tu y)$ is shifted for at least its entire length.
        \item Compute, bit-by-bit the disjunction of all the values above, which corresponds to concatenate $g(\tu y)$ and $h'(t,\tu y)$ because there is no carry (since \emph{at most} one disjunct is 1).
    \end{itemize}
\end{proof}
\end{toappendix}


To capture computation with counting gates or that goes beyond constant depth, one has to introduce more general schemas.

\begin{definition}
Given $g:\Nat^p \to \Nat$ and $h:\Nat^{p+1} \to \Nat$, the function $f:\Nat^{p+1} \to \Nat$ defined 
as the solution of the IVP with initial value $f(0,\tu y)=g(\tu y)$ and such that
\begin{itemize}
    \item $ 
    \frac{\partial f(x,\tu y)}{\partial \ell} = A(x, \tu y, h(x,\tu y)) \times f(x,\tu y) + B(x,\tu y, h(x, \tu y))
    $ 
    where $A,B\ge 1$ are $\fun{sg}$-polynomial expressions, is said to be defined by \emph{$\ell$-pODE}; 
    \item $ 
    \frac{\partial f(x, \tu y)}{\partial \ell} = - f(x, \tu y) + K\big(x, \tu y, h(x, \tu y), f(x, \tu y)\big)
    $ where $K\in \{0,1\}$ is a limited $\fun{sg}$-polynomial expression and $f$
    occurs in $K$ only under the scope of the sign function, is said to be defined by \emph{$\ell$-b$_0$ODE}; 
    %
    %
    \item $ 
    \frac{\partial f(x,\tu y)}{\partial \ell} = - f(x, \tu y) + B\big(x, \tu y, h(x,\tu y), f(x, \tu y)\big)
    $ 
    , where $B\ge 0$ is a $\fun{sg}$-polynomial expression, and $f$ 
    occurs in $B$ only in expressions of the form $\fun{sg}(f-c)$,
    for a constant $c$, is said to be defined by \emph{$\ell$-bODE}. 
\end{itemize}    
\end{definition}
Notably, $\ell$-pODE is the strict counterpart of the $\ell$-ODE schema introduced in~\cite{BournezDurand23} to capture $\FP$. In contrast, $\ell$-b$_0$ODE and $\ell$-bODE are non-strict schemas, allowing for (restricted) calls to $f$. By simply endowing $\ACDL$ with these schemas, it is proved that~\cite{ADK24a,ADK25}:
\begin{align*}
    \ACCt &= [\mathscr{B}_0; \circ, \ell\text{-ODE}_1, \ell\text{-b}_0\text{ODE}] \\
    \TC^0 &=[\mathscr{B}_0; \circ, \ell\text{-pODE}] \\
    \NC^1 &= [\mathscr{B}_0; \circ, \ell\text{-pODE}, \ell\text{-bODE}].
\end{align*}


Characterizations for $\ACCt$ and $\NC^1$ are then provided by simply endowing $\ACDL$ with these schemas:
\begin{align*}
    \ACCt &= [\mathscr{B}_0; \circ, \ell\text{-ODE}_1, \ell\text{-b}_0\text{ODE}] \\
    \NC^1 &= [\mathscr{B}_0; \circ, \ell\text{-pODE}, \ell\text{-bODE}].
\end{align*}

\section{Characterizing small circuit classes with modulo counting gates}\label{sec:algebras}

We extend the investigation started in~\cite{ADK24a,ADK25} by providing ODE-based function algebra characterizations for all $\ACCp{n}$  classes. 
Indeed, although $\ell$-b$_0$ODE, which intuitively corresponds to 1-BRN (i.e.~to $\ACCp{2}$), offers the advantage of being related to the (obviously, non-strict) schema characterizing $\NC^1$, it is not naturally generalizable to capture any $\ACCp{n}$. Instead we introduce simple schemas that naturally mirror modulo counting, allowing us to obtain uniform characterizations for all these classes. 
%
%
%
Consequently, our proofs are entirely self-contained and offer a marked simplification over existing ones (for $\ACCp{2}$ and $\ACCp{6}$)~\cite{CloteTakeuti}.
In Sec.~\ref{sec:ACCt}, we focus on $\ACCt$ as a case study, which, in Sec.~\ref{sec:ACCn}, is generalized by introducing a family of homogeneous schemas offering the very first characterizations for all $\ACCp{n}$. 
As a byproduct, $\TC^0$ is captured in a new way (Sec.~\ref{sec:TC}).
\arxiv{
\arxiv{Specifically, we will focus on computation performed through constant-depth circuits and} 
We show that, even when involving either general or modulo $n$ counting, \emph{strict} schemas are enough to obtain uniform characterizations: by carefully choosing coefficients $A$, $B$ in the defining \emph{linear} equation of the schema we provide ODE-based uniform characterizations for 
\arxiv{The main results provided are ODE-based characterizations for }classes ranging from $\AC^0$ to $\TC^0$. \arxiv{, surprisingly, obtained by carefully choosing coefficients $A$, $B$ in the defining \emph{linear} equation of the schema.}
Additionally, they are especially natural as based on simple schemas, which do not restrict existing schemas (for logarithmic depth computation) but naturally mirror (constant-depth) modulo counting\noarxiv{; while $\ell$-b$_0$ODE, which intuitively corresponds to Clote's 1-BRN, offers the advantage of being related with the (obviously, non-strict) schema for $\NC^1$, dealing with strict schemas somewhat reflects, in the ODE setting, the idea of constant-depth computation}.}
%
%
\arxiv{In Sec.~\ref{sec:ACCn} we introduce a family of uniform schemas naturally generalizing $\ell$-2ODE and provide very first characterizations for $\ACCp{n}$ (with $n$ possibly different from 2 and 6). 
As a byproduct we also capture $\TC^0$ in a new way (Sec.~\ref{sec:TC})
All results are investigated in both uniform and non-uniform settings.}

\subsection{A new characterization for $\ACCt$ by strict schemas}\label{sec:ACCt}

\arxiv{In this Section, we present characterizations for $\ACCt$ in both the uniform (Sec.~\ref{sec:unif}) and non-uniform settings (Sec.~\ref{sec:nonUnif}).
As seen, in~\cite{ADK25}, an ODE-based characterization for $\ACCt$ was already introduced but it was based on a non-strict schema.
While $\ell$-b$_0$ODE, which intuitively corresponds to Clote's 1-BRN, offers the advantage of being related with the (obviously, non-strict) schema for $\NC^1$, dealing with strict schemas somewhat reflects, in the ODE setting, the idea of constant-depth computation.  %
Additionally, the strict schema introduced in this section would, on the one hand, naturally scale to other strict schemas corresponding to computation by circuits with \textsc{Mod $n$} gates (see Sec.~\ref{sec:ACCn}) and, on the other, inspire the design of axioms to obtain \emph{natural} bounded arithmetic syntactically characterizing the corresponding classes (see Sec.~\ref{sec:bArit}).
\subsubsection{Characterizing $\ACCt$ in the uniform setting}\label{sec:unif}}
We start by introducing a strict schema defined as a special, limited case of $\ell$-ODE.
\begin{definition}[Schema $\ell$-2ODE]\label{df:tODE}
    Let $g:\Nat^p \to \{0,1\}$ and $k:\Nat^{p+1} \to \{0,1\}$. Then $f:\Nat^{p+1} \to \Nat$ is defined by $\ell$\emph{-2ODE} from $g$ and $k$ if it is the solution of the IVP s.t.\arxiv{with initial value} $f(0, \tu y) =g(\tu y)$ and \arxiv{such that:}
    $ 
    \frac{\partial f(x,\tu y)}{\partial \ell} = -2 \times k(x,\tu y) \times f(x,\tu y) + k(x,\tu y).
    $ 
\end{definition}

\noindent
%
Clearly, if $k(x,\tu y)=0$, $\frac{\partial f(x,\tu y)}{\partial \ell} = 0$, and if $k(x,y)=1$, $\frac{\partial f(x,\tu y)}{\partial \ell}=-2f(x,\tu y)+1$. Hence,  $f(x,\tu y)$ takes values in $\{0,1\}$,
that, when $k$ is positive, switch from $0$ to $1$, and vice versa. 
%
As desired the computation through $\ell$-2ODE is in $\ACCt$.
%
\begin{proposition}\label{prop:2ODE}
    If $f$\arxiv{$f(x,\tu y)$} is defined by $\ell$-2ODE from functions in $\ACCt$, $f$ is in $\ACCt$. \arxiv{$f(x,\tu y)$ is in $\ACCt$ as well.}
\end{proposition}
\begin{proof}[Proof sketch]
    By Def.~\ref{df:tODE}, $f(x, \tu y)=\sum^{\ell(x)-1}_{u=-1}k(\alpha(u),\tu y)\times (-1)^{\# \{ t \in [u+1, \ell(x)-1]:k(\alpha(t), \tu y)=1\}}$. 
    All terms in which $k(\alpha(u), \tu y)=0$ vanish from the sum.
    \arxiv{from the sum (recall that $k(\alpha(-1), \tu y)=g(\tu y))$.}
    Let $\{u_1, \dots, u_a\}=\{t\in [-1, \ell(x)-1]: k(\alpha(t), \tu y)=1\}$. Then,
    $f(x,\tu y)=\sum^{a}_{i=0}(-1)^{a-i}$, which is 0 if $a$ is odd and 1 otherwise.
\end{proof}

\begin{toappendix}
\subsection{Proofs from Section~\ref{sec:ACCt}}

\begin{proof}[Proof of Prop.~\ref{prop:2ODE}]
    By Def.~\ref{df:tODE}, 
    $$
    f(x, \tu y) = \sum^{\ell(x)-1}_{u=-1} \bigg(\prod^{\ell(x)-1}_{t=u+1} \Big(1-2k(\alpha(t), \tu y)\Big)\bigg) \times k(\alpha(u), \tu y)
    $$
    for $k\leq 1$ and with the convention that $k(\alpha(-1), \tu y)=f(0,\tu y)$.
    This is computable in $\ACCt$.
    Indeed, it holds that:
    $$
    f(x,\tu y) = \sum^{\ell(x)-1}_{u=-1} k(\alpha(u), \tu y) \times (-1)^{\sharp \{t \in [u+1, \ell(x)-1]:k(\alpha(t),\tu y)=1\}}.
    $$
    All terms above in which $k(\alpha(u),\tu y)=0$ vanish from the sum (recall that $k(\alpha(-1), \tu y)=g(\tu y)$).
    Let $\{u_1, \dots, u_a\} = \{t\in [-1, \ell(x)-1] : k(\alpha(t), \tu y)=1\}$.
    Then,
    $$
    f(x,\tu y) = \sum^a_{i=0}(-1)^{a-i} = \sum_{i=0}^a (-1)^i = \begin{cases}
        0 \quad &\text{if } a \text{ is odd} \\
        1 \quad &\text{if } a \text{ is even.}
    \end{cases}
    $$
    In terms of (constant-depth) circuit, to compute $f(x,\tu y)$ we first compute all values $g(\tu y)$ and $k(\alpha(u),\tu y)$ for $u\in \{0,\dots, \ell(x)-1\}$.
    Then, we count the parity of those who are equal to 1 by a \textsc{Mod 2} gate.
    Since $g$ and $k$ are in $\ACCt$ by hypothesis, the whole computation is also in $\ACCt$. 
\end{proof}
\end{toappendix}


\noindent
Remarkably, this schema also offers a natural tool to compute parity, for example, by checking whether the number of 1s in the binary representation of its input is even or odd. This function is especially crucial, as it provides one of the few known separation results between small circuit classes (namely, separating $\AC^0$ and $\ACCt$).

\begin{remark}[Counting modulo 2]\label{remark:tODE}
    Computation performed by $\textsc{Mod 2}$ gates can be simulated via $\ell$-2ODE by considering the special case of $g(x, y) =0$ and $k(x,\tu y) =\fun{BIT}(\ell(x), y)$, i.e. by considering the solution of the IVP below:
    \begin{align*}
        \fun{cmod2}(0, y) &= 0 \\
        \frac{\partial \fun{cmod2}(x, y)}{\partial \ell} &= -2\fun{BIT}(\ell(x), y) \times \fun{cmod2}(x, y) + \fun{BIT}(\ell(x),  y). 
    \end{align*}
    Then, the count of input bits modulo 2 is obtained by computing $\fun{cmod2}(x,x)$.
\end{remark}

\begin{example}
    Let $n=2$, so that $z=z_1z_0$. Then,
    \begin{align*}
        \fun{cmod2}(0, z) &= 0 \\
        \fun{cmod2}(1,z) &= \fun{cmod2}(0,z) - 2\fun{BIT}(0,z) \times \fun{cmod2}(0,z) + \fun{BIT}(0,z)= \fun{BIT}(0,z) = z_0 \\
    %
        \fun{cmod2}(2,z) &= \fun{cmod2}(1,z) - 2\fun{BIT}(1,z) \times \fun{cmod2}(1,z) + \fun{BIT}(1,z) = z_0 - 2z_0z_1 + z_1 \\
        &= \begin{cases}
            0 \quad &\text{if } z_0=z_1 \\
            1 \quad &\text{otherwise.}
        \end{cases}
    \end{align*}
\end{example}

\noindent
Therefore, to obtain full characterization of $\ACCt$ we simply endow $\ACDL$ with the $\ell$-2ODE schema, capturing \textsc{Mod 2} gates.
\begin{theorem}\label{theorem:ACCt}
    $\ACCt = [\mathscr{B}_0; \circ, \ell\emph{-ODE}_1, \ell\emph{-2ODE}]$.
\end{theorem}
\begin{proof}[Proof Sketch]
    $(\supseteq)$ As in~\cite{ADK24a} for functions and schemas in $\ACDL$, plus Prop.~\ref{prop:2ODE} for $\ell$-2ODE.
    $(\subseteq)$ By capturing computation by circuits defining $\ACCt$ via $\ell$-2ODE (Remark~\ref{remark:tODE}). 
\end{proof}

\begin{toappendix}
\begin{proof}[Proof of Theorem~\ref{theorem:ACCt}]
    $(\subseteq)$ For all functions and schemas in $\ACDL$ the proof is analogous to the one presented in~\cite{ADK24a}, while the closure of $\ACCt$ under $\ell$-2ODE has been established in Prop.~\ref{prop:2ODE}.
    $(\supseteq)$ Recall that computation through constant-depth, polynomial-size circuits is already captured by $\ACDL$.
    Moreover, by Remark~\ref{remark:tODE}, $\ell$-2ODE allows to capture computation performed by \textsc{Mod 2} gates, the only gates to be added to $\AC^0$ circuits to obtain $\ACCt$.
    This is enough to conclude our proof (see, e.g.,~\cite[Th. 2.2]{CloteTakeuti}).
\end{proof}
\end{toappendix}

\noindent
%
The relationship between strict $\ell$-2ODE and non-strict $\ell$-b$_0$ODE is clarified by the remark below, which, due to closure of $\ell$-b$_0$ODE under $\ACCt$ computation~\cite[Th. 19]{ADK25}, offers an alternative proof of Prop.~\ref{prop:2ODE}.

\begin{remark}[$\ell$-2ODE vs.~$\ell$-b$_0$ODE]\label{remark:2ODEbODE}
    The schema $\ell$-2ODE can be seen as a special case of $\ell$-b$_0$ODE s.t.~$K(x,\tu y, h_k(x,\tu y), f(x,\tu y))=\fun{cosg}(f(x,\tu y))\times h_k(x,\tu y)+\fun{sg}(f(x,\tu y))\times \fun{cosg}(h_k(x,\tu y))$.
\end{remark}

\begin{toappendix}
\begin{proof}[Proof of Remark~\ref{remark:2ODEbODE}]
\begin{sloppypar}
    The schema $\ell$-2ODE can be seen as 
    a special case of $\ell$-b$_0$ODE such that $K(x,\tu y, h_k(x,\tu y), f(x,\tu y))=\fun{cosg}(f(x,\tu y))\times h_k(x,\tu y)+\fun{sg}(f(x,\tu y))\times \fun{cosg}(h_k(x,\tu y))$, i.e.,~a function $f$ defined by $\ell$-2ODE from $g_k$ and $h_k$ can be rewritten as an instance of $\ell$-b$_0$ODE with initial value $f(0,\tu y)=g_k(\tu y)$ and such that:
    \begin{align*}
        \frac{\partial f(x,\tu y)}{\partial \ell} = -f(x,\tu y) + \fun{cosg}(f(x,\tu y)) \times h_k(x,\tu y) 
        + \fun{sg}(f(x,\tu y)) \times \fun{cosg}(h_k(x,\tu y)).
    \end{align*}
    Clearly, since it is easily provable that $f(x,\tu y)$ returns only values in $\{0,1\}$, this can be rephrased as follows:
    \small
    \begin{align*}
    f(x,\tu y) &= f(\ell(x)-1) - f(\ell(x)-1) + \begin{cases}
         \fun{cosg}(h_k(\alpha(x)-1)) \; &\text{if } \fun{sg}(f(\alpha(\ell(x)-1))=1 \\
        h_k(\ell(x)-1) \; &\text{if } \fun{cosg}(f(\ell(x)-1))=1
    \end{cases} \\
    &= \begin{cases}
        \fun{cosg}(h_k(x,\tu y)) \quad \quad \quad \quad \quad \quad \quad \quad \quad \quad \quad \quad \quad \quad   \; \; &\text{if } f(\ell(x)-1)  =1 \\
        h_k(x,\tu y) \quad &\text{if } f(\ell(x)-1) =0
    \end{cases} \\
    &= \begin{cases}
        0 \quad \quad \quad \quad \quad \quad \quad \quad \quad \quad \quad \quad  \quad \quad \quad \quad \quad \quad \quad \quad \;
        &\text{if } h_k(\ell(x)-1) = f(\ell(x)-1)  \\
        1 \quad &\text{if } h_k(\ell(x)-1) 
        \neq f(\ell(x)-1) 
    \end{cases} \\
    &= f(\alpha(\ell(x)-1), \tu y)  - 2 f(\alpha(\ell(x)-1), \tu y) \times h_k(\ell(x)-1) + h_k(\ell(x)-1).
    \end{align*}
    \normalsize
    where $f(\ell(x)-1)$ is a shorthand for $f(\alpha(\ell(x)-1), \tu y)$ and similarly $h_k(\ell(x)-1)$ for $h_k(\alpha(\ell(x)-1),\tu y)$.
    Due to completeness via $\ell$-b$_0$ODE~\cite[Th.~19]{ADK25}, this simple rewriting would offer an alternative proof of Prop.~\ref{prop:2ODE}.
\end{sloppypar}
\end{proof}
\end{toappendix}

\noindent
More general connections can be established between the 1-BRN, $\ell$-b$_0$ODE and $\ell$-2ODE schemas (see Remark~\ref{remark:2ODEvs1BRN} in App.~\ref{app:ACCp} for the former two and giving a third, indirect, proof of Th.~\ref{theorem:ACCt}).
However, using $\ell$-2ODE offers the advantage of a very straightforward proof that does not rely on results in group theory or convoluted combinatorial arguments.  Additionally, its close link with modulo counting resources is what enables us to scale the characterization across all classes; unlike $k$-BRN which jumps to $\NC^1$ already for $k=4$. 

\raus{
While alternative known schemas to capture this class, namely,~1-BRN~\cite{CloteTakeuti} and $\ell$-b$_0$ODE~\cite{ADK25}, can be put in connection to ours, this characterization through $\ell$-2ODE offers the advantage of a very straightforward proof; in contrast, equivalent characterizations proposed to date either rely on results in group theory or present rather convoluted combinatorial arguments.
The fact that our strict recursive schema is more closely related to constant-depth (counting) circuit  resources is what enables us to scale the characterization across all classes; this is a result that previous ``ad hoc'' schemas, restricting recursion originally introduced for log-depth circuits, have failed to deliver:
%
%
%
differently from $\ell$-2ODE, 1-BRN (and 2-BRN) are in some sense unnatural to capture constant-depth circuit computation directly.
    It can be shown that 1-BRN can be rewritten as an instance of $\ell$-2ODE (see Remark~\ref{remark:2ODEvs1BRN} in App.~\ref{app:FACCt}), obtaining a third (indirect) proof of Th.~\ref{theorem:ACCt}.
} 

\begin{toappendix}

The connection between 1-BRN and $\ell$-2ODE is clarified by the following remark, showing that the former can be rewritten as an instance of the latter.

\begin{remark}[1-BRN vs.~$\ell$-2ODE]\label{remark:2ODEvs1BRN}
    Recall that a function $f(x,\tu y)$ is defined by 1-BRN from $g$ and $h$ if it such that
    $
    f(0, \tu y) = g(\tu y)$
    and 
    $f(x,\tu y) = h\big(x,\tu y, f\big(\big\lfloor \frac{x}{2}\big\rfloor, \tu y\big)\big)
    $
    where, for any $x,\tu y$, $f(x, \tu y) \leq 1$ and $i\in \{0,1\}$.
    Since $f(x, \tu y)$ takes only values 0 or 1, $f(x, \tu y) = f\big(\big\lfloor \frac{x}{2}\big\rfloor, \tu y\big) \times \big(h(x, \tu y, 1)- h(x,\tu y,0)\big) +  h(x, \tu y, 0)$.
    This can be rewritten as the solution of an IVP defined by $\ell$-2ODE from the function $g$ above and
    $$
    k(x, \tu y) = \begin{cases}
        0 \quad &\text{if } h_0(x) \neq h_1(x) \text{ and } h_0(x) = 0 \\\
        &\text{or } h_0(x) = h_1(x), \fun{changes}(x)= i 
        \text{ and } h_0(x)=i \\
        1 \quad &\text{otherwise}
    \end{cases}
    $$
    \begin{sloppypar}
    \noindent where $h_i(x)$ is a shorthand for $h(x, \tu y, i)$ and $\fun{changes}(x)=i$ abbreviates $\mu.z \in \{x-1,\dots, 0\} (h_0(z)=h_1(z))=p \wedge \fun{cmod2}\big(\sum^{x-1}_{j=p} h_0(j) \neq h_0(j+1)\big)=i$, intuitively returning the number of times such that $h_0(t)$ (and $h_1(t)$) changes with respect to the preceding value $h_0(t-1)$ between $x$ (for which $h_0(x)=h_1(x)$) and the closest value $x'$ such that $h_0(x')\neq h_1(x')$.
    \arxiv{For further details, see Appendix~\ref{app:FACCt} below.}
    \end{sloppypar}
\end{remark}

\noindent
In~\cite{ADK25}, it is shown that endowing $\mathcal{A}_0$ (i.e.,~$\ACDL$) with $\ell$-b$_0$ODE is enough to capture $\ACCt$.
Since all the functions needed to rewrite 1-BRN using $\ell$-2ODE are in the algebra introduced in Th.~\ref{theorem:ACCt}, Remark~\ref{remark:2ODEvs1BRN} provides a third, alternative (indirect~\cite{CloteTakeuti}) proof of completeness for the desired characterization.

\arxiv{
\subsubsection{Relating $\ell$-2ODE, 1-BRN and $\ell$-b$_0$ODE}\label{app:FACCt}
(1-BRN $\Longrightarrow$ $\ell$-2ODE)
If $h_0(x) \neq h_1(x)$ and $h_0(x)=0$ ($h_1=1$), the value of $f(x,\tu y)$ does not change with respect to the previous one, i.e.~$f(x,\tu y)=f\big(\big\lfloor \frac{x}{2}\big\rfloor, \tu y\big)$.
Accordingly, in this case $k(x,\tu y)=0$.
Similarly, if $h_0(x)\neq h_1(x)$ and $h_0(x)=1$ (so $h_1(x)=0$), then the value of $f(x,\tu y)$ must change with respect to the previous one, i.e.,~$f(x,\tu y)=1-f\big(\big\lfloor \frac{x}{2}\big\rfloor, \tu y\big)$.
In this case, $k(x,\tu y)=1$, so that
$$
\frac{\partial f(x,\tu y)}{\partial \ell} = -2 f(x, \tu y) + 1 = \begin{cases}
0 \quad &\text{if } f(\alpha(\ell(x)-1), \tu y) = 1 \\
1 \quad &\text{if } f(\alpha(\ell(x)-1),\tu y) = 0
\end{cases}
$$
On the other hand, if $h_0(x)=h_1(x)$, we need to look at \emph{the first} $t\in\{0,\dots, x\}$ such that $h_0(t)\neq h_1(x)$.
This first value can be found via sharply bounded minimization, which is in $\ACDL$ (see Remark~\ref{remark:minimisation}).
As seen, for any $z$, when $h_0\big(\big\lfloor \frac{z}{2}\big\rfloor\big)\neq h_1\big(\big\lfloor\frac{z}{2}\big\rfloor\big)$, if $h_0\big(\big\lfloor \frac{z}{2}\big\rfloor\big)=0$ (and $h_1\big(\big\lfloor \frac{z}{2}\big\rfloor\big)=1$), the value of $f\big(\big\lfloor \frac{z}{2}\big\rfloor, \tu y\big)$ is preserved (i.e.,~$f(z,\tu y)=f\big(\big\lfloor \frac{z}{2}\big\rfloor, \tu y\big)$),
while if $h_0\big(\big\lfloor \frac{z}{2}\big\rfloor\big)=1$, the value of $f\big(\big\lfloor \frac{z}{2}\big\rfloor, \tu y\big)$ changes.
Clearly, since $f(x,\tu y)\leq 1$, two changes correspond to going back to the original value.
Then, since $t$ is the first value (with respect to the input $x$) such that $h_0(t)\neq h_1(t)$, we know that, in all other cases $t'\in \{t,\dots, x\}$, $h_0(t')=h_1(t')$.
So, we only need to check if $h_0(t)=0$ \big(i.e.,~$f(t, \tu y)=f\big(\big\lfloor \frac{t}{2}\big\rfloor\big)$\big) or $h_0(t)=1$ (i.e.,~$f(t, \tu y)=1-f\big(\big\lfloor \frac{t}{2}\big\rfloor, \tu y\big)$) and count modulo 2 how many times the value of $f$ changes, i.e.~$h_0(t')\neq h_0\big(\big\lfloor \frac{t'}{2}\big\rfloor\big)$.
This can be done in $\ACDL$ using $\fun{cmod2}$ (see Remark~\ref{remark:2ODE}).
In particular, $f(x,\tu y)= f\big(\big\lfloor\frac{t}{2}\big\rfloor, \tu y\big)$, when either (i) $h_0(x)=0$ and the number of changes modulo 2 is 0 or (ii) $h_0(x)=1$ and the number of changes modulo 2 is 1.

$(\ell$-2ODE $\Longrightarrow$ 1-BRN) Any $f(x,\tu y)$ defined by $\ell$-2ODE from $g$ and $k$ can be defined by 1-BRN from the same $g$ and $h_k$, where $h_k(x,\tu y, 0)= k(x,\tu y)$ and $h_k(x,\tu y, 1)=1-k(x,\tu y)$.

\begin{sloppypar}
(1-BRN $\Longrightarrow$ $\ell$-b$_0$ODE)
If $f$ is defined by 1-BRN from $g$ and $h$, then it can be rewritten (directly) by $\ell$-b$_0$ODE considering $K(x,\tu y, f(x,\tu y), h(x,\tu y))=\fun{cosg}(f(x,\tu y))\times h(x,\tu y, 0) +\fun{sg}(f(x,\tu y))\times h(x,\tu y,1)$.
\end{sloppypar}

($\ell$-2ODE $\Longrightarrow \ell$-b$_0$ODE) See Remark~\ref{remark:2ODE}.}
\end{toappendix}

We conclude our investigation of $\ACCt$ by remarking that similar completeness results holds in the non-uniform setting too. 
In this context, functions and relations that described the circuits for all $n$ (the so-called direct connection language~\cite{Vollmer}) are given as basic functions and are added to $\mathscr{B}_0$. 

\begin{proposition}\label{prop:nonUnif}
    If $f:\Nat \to \Nat$ is computable by a family $\mathcal{C} = (C_n)_{n\ge 0}$ of polynomial size and constant depth circuits including \textsc{Mod 2} gates, then it is in 
    $
    [\mathscr{B}_0, \mathbf{circ}_{\mathcal{C}}; \circ, \ell\emph{-ODE}_1, \ell\emph{-2ODE}],
    $
    where $\mathbf{circ}_{\mathcal{C}}=\{C,L_0^{in}, L_0^{\neg}, L_e\}$, denotes the set of characteristic functions associated to predicate and relations describing the underlying graph of a circuit and its level of gates, resp.
\end{proposition}

\raus{
\begin{proof}[Proof Sketch]
    Following the construction of~\cite[Prop.~3.28]{ADK24a} except for the function $\fun{Eval}$ which here includes inductive levels corresponding to layers computing \textsc{Mod 2}: $\fun{Eval}_{3e}(t,x)$ is $\fun{cmod2}\Big(\sum^{\ell(x)^k}_{i=0} \{\fun{Eval}_{3e-3}(i,x) : C(x,i,t)=1\}\Big)$, i.e.~counting modulo 2 a number of bits which is at most polynomial in the length of the input and expressible by $\ell$-2ODE (Remark~\ref{remark:2ODE}).
\end{proof}
}

\begin{toappendix}
\subsubsection{Characterizing $\ACCt$ in the non-uniform setting}\label{sec:nonUnif}

%
%
Let $\mathcal{C}=(C_n)_{n\ge 0}$ be the class of circuits of polynomial size $n^k$, for $k\in \Nat$, and constant depth $d$, including \textsc{Mod 2} gates.
We assume that each circuit $C_n$ is in a special normal form, such that it strictly alternates between layers of $\wedge,\vee$ and \textsc{Mod 2} and edges are only between gates of consecutive layers:
input gates and their negations are all at level 0, $\vee$ gates are at levels 3$e$ + 1, $\wedge$ gates are at levels 3$e$ + 2 and \textsc{Mod 2} gates are at levels $3e$.
As in~\cite{ADK24a}, we keep all the basic functions of $\ACDL$, but we add a set $\mathbf{circ}_\mathcal{C} = \{C, L_0^{in}, L_0^\neg, L_e\}$ of characteristic functions associated to the predicate $C\subseteq \Nat^n$ and relations $L_0^{in}, L_0^\neg, L_e \subseteq \Nat \times \Nat$, for $e\in\{1,\dots, d\}$ .
Intuitively, $C$ describes the underlying graph of the circuit:
for any integers $x,\alpha, \beta, (x,\alpha,\beta)\in C$, when in $C_{\ell(x)}$ the $\alpha^{th} \leq \ell(x)^{k}$ gate of some level is a predecessor of the $\beta^{th} \leq \ell(x)^k$ gate of the next level.
Thus, in this encoding, $\alpha$ and $\beta$ are exponentially smaller than $x$.
Relations $L_0^{in}, L_0^\neg$ and $L_e$ describe the level of gates (and, implicitly, their type):
$L_0^{in}$ refers to input gates, $L_0^{in}$ to negation gates, and $L_e$ to $\wedge,\vee$ or \textsc{Mod 2} gates, depending on $e$.
Since we aim to define functions over integers, we assume that input gates are numbered from $n-1$ to 0, and output gates range from $m-1$ to 0, with $m\leq \ell(x)^k$.
It is by considering the function corresponding to the given relations that we obtain the desired ODE style family of classes (parametrized on $\mathcal{C}$).

\begin{proof}[Proof of Prop.~\ref{prop:nonUnif}]
The proof follows the construction described in~\cite{ADK24a}, except for the function $\fun{Eval}$ which is defined as follows.
Let $\fun{Eval}(t,x)$ be a function that returns the value of the $t^{th}$ output gate of the circuit $C_{\ell(x)}$ of input $x$ when $t\leq m-1$ and 0 otherwise:
$$
\fun{Eval}(t,x) = \begin{cases}
    C_{\ell(x)} \quad &\text{if } t \leq m-1 \\
    0 \quad &\text{otherwise.}
\end{cases}
$$
To define the function $\fun{Eval}(t,x)$, we assume that $C_n$ has depth $d$.
We first define a special sharply bounded minimum operator function that, given $k\in \Nat$ and two functions $g$ and $h$, with $h(t,\tu y)\leq 1$, $t\in \Nat$ and $\tu x=x_1, \dots, x_h$, 
$$
\fun{min}_{i\leq \ell(x)^k}\{g(i,\tu x) : h(i,\tu x) \triangleright j\}
$$
\begin{sloppypar}
\noindent
with $\triangleright \in \{<,\leq, >, \ge,=\}$ and $j\in\{0,1\}$.
Intuitively, given $\{0, \dots, \ell(x_1)\}$, it computes the minimum of the values of $g(i,\tu x)$ for $i$ and $\tu x$ such that $h(i,\tu x)\triangleright j$.
\end{sloppypar}
The inductive definition of $\fun{Eval}$ relies on these of $d+1$ functions $\fun{Eval}_0, \dots, \fun{Eval}_d$ with $\fun{Eval}_d=\fun{Eval}$:
\begin{itemize}
    \itemsep0em
    \item $\fun{Eval}_0(t,x)$ is equal to $\fun{BIT}(t,x)$ if $L_0^{in}(t,x)$ holds and $1-\fun{BIT}(t,x)$ if $L_0^\neg(t,x)$ does.
    For $t$ not corresponding to gate index, $\fun{Eval}_0(t,x)$ is set to an arbitrary value, say 0.
    \begin{sloppypar}
    \item $\fun{Eval}_{3e+1}(t,x)$ is equal to $\fun{min}_{i\leq \ell(x)^k}\{\fun{Eval}_{3e}(i,x) : C(x,i,t) = 1\}$, for $L_{3e}$, i.e.,~if $t$ is the index of a gate at this level.
    The evaluation of the $i^{th}$ gate of the level $2e$ (a $\wedge$-gate) is the minimum of the evaluations of its predecessors gates of level $3e$.
    Again, this function is already in non-uniform $\ACDL_{\mathcal{C}}$, i.e.,~the class considered here but without $\ell$-2ODE (see~\cite{ADK24a}).
    \end{sloppypar}
    \begin{sloppypar}
    \item Analogously, $\fun{Eval}_{3e+2}(t,x)$ is the $1-\fun{min}_{i\leq \ell(x)^k}\{1-\fun{Eval}_{3e+1}(i,x) : C(x,i,t)=1\}$.
    The evaluation of the $t^{th}$ gate of level $3e+2$ (a $\vee$-gate) is the maximum among evaluations of its predecessor gates at level $2e$.
    Again, this is already in non-uniform $\ACDL_{\mathcal{C}}$.
    \item $\fun{Eval}_{3e}(t,x)$ is equal to $\fun{cmod2}\big(\sum^{\ell(x)^k}_{i=0}\{\fun{Eval}_{3e-3}(i,x):C(x,i,t)=1\}\big)$.
    As seen, this corresponds to counting modulo 2 a number of bits which is at most polynomial in the length of the input and can be defined by $\ell$-2ODE (see Remark~\ref{remark:2ODE}).
    \end{sloppypar}
\end{itemize}
Finally, we consider the following expression, which defines a function $f$ such that $f(2^{\ell(x)^k},x)$ outputs the value of the computation of $C_n$ (for $n=\ell(x)$) on input $x$:
$$
\frac{\partial f(y,x)}{\partial \ell(y)} = f(y,x) + \fun{Eval}(\ell(x)-\ell(y)-1, x)
$$
with initial value $f(0,x)=0$.
Intuitively, the given function computes the successive suffixes of the output word, starting from the bit of the bigger weight.
Remarkably, this is an instance of $\ell$-ODE$_1$ (as $\fun{Eval}(y,x) \in \{0,1\}$).
So, the given $f$ can be rewritten in the desired class.
\end{proof}

\end{toappendix}


\subsection{A family of algebras for $\ACCp{n}$}\label{sec:ACCn}

\arxiv{
In this Section, we introduce a family of uniform schemas naturally generalizing $\ell$-2ODE and extend the results of Sec.~\ref{sec:ACCt} to all $\ACCp{n}$.
Relying on ODE-schemas, we not only obtain the very first characterizations for $\ACCp{n}$ (with $n$ possibly different from 2 and 6) and $\cc{FACC^0}$, but we also capture $\TC^0$ in a new way (Sec.~\ref{sec:TC}).
We first consider computation through circuits with \textsc{Mod 3} gates and then show how the approach naturally generalizes to the study of circuits including \textsc{Mod $n$}, for any $n\in \Nat$ (either prime or composite), and to \textsc{Maj} gates.}

\begin{toappendix}
\subsubsection{From $\ACCp{3}$ to $\ACCp{n}$}\label{app:ACCp}
\begin{sloppypar}
The schema characterizing $\ACCp{3}$ is a special case of $\ell$-ODE whose defining equation involves integer division, i.e.~$\big\lfloor \frac{f(x,\tu y)}{2}\big\rfloor$, $A=-3k(x,\tu y), B=k(x,\tu y)$ and $k(x,\tu y)\leq 1$.\footnote{This schema can be seen as a generalization of $\ell$-ODE$_3$, see~\cite{ADK24a}.}
\end{sloppypar}
\begin{definition}[$\ell$-3ODE schema]
    Let $g:\Nat^p \to \{0,1,2\}$ and $k:\Nat^{p+1}\to \{0,1\}$, then $f:\Nat^{p+1}\to \Nat$ is defined by $\ell$\emph{-3ODE} from $g$ and $k$ if it is the solution of the IVP with initial value $f(0,\tu y)=g(\tu y)$ and such that:
    $$
    \frac{\partial f(x,\tu y)}{\partial \ell} = - 3 \times k(x,\tu y) \times \bigg\lfloor \frac{f(x,\tu y)}{2} \bigg\rfloor + k(x,\tu y).
    $$
\end{definition}

\noindent
Clearly, either $k(\alpha(\ell(x)-1),\tu y)=0$ and $f(x+1,\tu y)=f(\alpha(\ell(x)-1),\tu y)$ or $k(\alpha(\ell(x)-1),\tu y)=1$ and then,
\begin{align*}
    f(x, \tu y) &= f(\alpha(\ell(x)-1), \tu y) - 3 \times \bigg\lfloor \frac{f(\alpha(\ell(x)-1),\tu y)}{2} \bigg\rfloor + 1 \\
    &= \begin{cases}
        f(\alpha(\ell(x)-1), \tu y) + 1 \quad &\text{if } f(\alpha(\ell(x)-1), \tu y) < 2 \\
        f(\alpha(\ell(x)-1), \tu y)-2 = 0 \quad &\text{if } f(\alpha(\ell(x)-1), \tu y)=2.
    \end{cases}
\end{align*}

\begin{proposition}\label{prop:3ODE}
    If $f$ is defined by $\ell$\emph{-3ODE} from functions in $\ACCp{3}$, then $f$ is in $\ACCp{3}$ as well.
\end{proposition}

\begin{proof}
    Straightforward consequence of the observation above. 
    Indeed, $f(x,\tu y)$ can be computed by a circuit including \textsc{Mod 3} gates in constant depth as follows:
    \begin{itemize}
        \item In the beginning, compute $g(\tu y)$ and $k(\alpha(t), \tu y)$'s for any $t\in \{0,\dots, x-1\}$, independently and in parallel.
        This can be done in $\ACCp{3}$ by hypothesis.
        \item In one step, compute the sum of these values modulo 3 using one \textsc{Mod 3} gate.
    \end{itemize}
\end{proof}

That this schema is enough to capture $\ACCp{3}$ is established following the techniques of Remark~\ref{remark:tODE}.

\begin{remark}[Counting modulo 3]\label{remark:cmod3}
    Computation performed by \textsc{Mod 3} gates can be simulated by $\ell$-3ODE from $k$, for $k(x,\tu y)=\fun{BIT}(x,\tu y)$; that is,~by computing $\fun{cmod3}(x,x)$, where $\fun{cmod3}(x,\tu y)$ is defined as the solution of the IVP with initial value $\fun{cmod3}(0,\tu y)=0$ and such that
    $
    \frac{\partial \fun{cmod3(x,\tu y)}}{\partial \ell} = - 3 \times \bigg\lfloor \frac{\fun{cmod3}(x,\tu y)}{2}\bigg\rfloor \times \fun{BIT}(\ell(x), \tu y) + \fun{BIT}(\ell(x), \tu y).
    $
    Clearly, for any $t\in\{0,\dots, x\}$ and $\tu y$, $\fun{cmod3}(\alpha(t),\tu y)\leq 2$.
    In particular, whenever $\fun{comd3}(t,\tu y) < 2$, the next value $t' \in \{0,\dots, x\}$ such that $\ell(t')=\ell(t)+1$ will be simply
    $
    \fun{cmod3}(t', \tu y) = \fun{cmod3}(t,\tu y) + \fun{BIT}(t, \tu y),
    $
    while when $\fun{cmod3}(t, \tu y)=2$, 
    $$
    \fun{cmod3}(t', \tu y) = 
    \begin{cases}
       \fun{cmod3}(\alpha(t), \tu y) = 2 &\text{if } \fun{BIT}(\alpha(t), \tu y) = 0 \\
       \fun{cmod3}(\alpha(t),\tu y)-2 = 0 &\text{if } \fun{BIT}(\alpha(t), \tu y)=1
    \end{cases}
    $$
\end{remark}

\noindent
Then, $\ACCp{3}$ can be captured by simply endowing $\ACDL$ with $\ell$-3ODE.
Completeness in the uniform and non-uniform setting is established exactly as before.

\begin{lemma}\label{lemma:3ODE}
    $\ACCp{3}=[\mathscr{B}_0; \circ, \ell$\emph{-ODE}$_1$, $\ell$\emph{-3ODE}$]$.
\end{lemma}

\begin{proof}[Proof of Lemma~\ref{lemma:3ODE}]
$(\supseteq)$ As in~\cite{ADK24a}, for all functions and schemas except $\ell$-3ODE, the closure property of which is proved in Proposition~\ref{prop:3ODE}.
$(\subseteq)$ By~\cite{ADK24a} plus Remark~\ref{remark:cmod3}, capturing computation through \textsc{Mod 3} gates.
\end{proof}

\begin{lemma}\label{lemma:3nonUnif}
\begin{sloppypar}
If $f:\Nat \to \Nat$ is computable by a family of polynomial size and constant depth circuits including \textsc{Mod 3} gates, $\mathcal{C}=(C_n)_{n\ge 0}$, then $f$ is in the class $[\mathscr{B}_0, \mathbf{circ}_{\mathcal{C}}; \circ, \ell\emph{-ODE}_1, \ell\emph{-3ODE}]$.
\end{sloppypar}
\end{lemma}

\begin{proof}[Proof of Lemma~\ref{lemma:3nonUnif}]
Precisely as for Prop.~\ref{prop:nonUnif} but constructing $\fun{Eval}$ so that $\fun{Eval}_{3e}(y,x)$ is equal to $\fun{cmod3}\Big(\sum^{\ell(x)^k}_{i=0} \{\fun{Eval}_{3e-3}(i,x):C(x,i,t)=1\}\Big)$.
\end{proof}
\end{toappendix}

Nicely, the framework of Sec.~\ref{sec:ACCt} uniformly scales to all natural numbers $n\in \Nat$, when considering the family of schemas $\ell$-$n$ODE below.

\begin{definition}[$\ell$-$n$ODE schema]\label{def:nODE}
    Let $n\in\Nat^{>1}$, $g:\Nat^{p}\to \{0,\dots, n-1\}$ and $k:\Nat^{p+1} \to \{0,1\}$, then $f:\Nat^{p+1}\to \Nat$ is defined by $\ell$-$n$ODE from $g$ and $k$ if it is the solution of the IVP with initial value $f(0,\tu y)=g(\tu y)$ and s.t.:
    $$
        \frac{\partial f(x,\tu y)}{\partial \ell} = - n\times k(x,\tu y) \times \bigg\lfloor \frac{f(x,\tu y)}{n-1} \bigg\rfloor + k(x,\tu y).
    $$
\end{definition}
\noindent
Note that the integer division symbol is used here with a slight abuse of notation since, for any $x$ and $\tu y$, $f(x,\tu y)$ takes a value smaller than $n$, so that division is actually only checking whether $f(x,\tu y)=n-1$.
Hence, the kind of division involved here can be computed by circuits of constant depth.
Indeed, if $k(x,\tu y)=0$, then $f(x, \tu y)=f(\alpha(\ell(x)-1),\tu y)$, while if $k(x,\tu y)=1$, then $f(x, \tu y)= f\big(\alpha(\ell(x)-1), \tu y\big)- n\times \Big \lfloor \frac{f(\alpha(\ell(x)-1),\tu y)}{n-1}\Big\rfloor + 1$, i.e.~$f(x,\tu y)= f\big(\alpha(\ell(x)-1), \tu y\big)+1$ if $f\big(\alpha(\ell(x)-1), \tu y\big) < n-1$ and $f(x,\tu y)=f\big(\alpha(\ell(x)-1), \tu y\big) - n = 0$ if $f(\alpha(\ell(x)-1), \tu y)=2$. 
From this the closure property follows (see App.~\ref{app:ACCp}).

\begin{proposition}\label{prop:nODE}
    Let $n\in \Nat$.
    If $f$ is defined by $\ell$-$n$\emph{ODE} from functions in $\ACCp{n}$, then $f$ is in $\ACCp{n}$ as well.
\end{proposition}

\noindent
This bound on $f$ values  also enforces that the typical $\TC^0$ function \textsc{BCount}~\cite{Vollmer} (counting the number of 1's of the binary representation of its input) cannot be defined by $\ell$-$n$ODE. 

\begin{toappendix}
\begin{proof}[Proof of Prop.~\ref{prop:nODE}]
By generalizing Prop.~\ref{prop:2ODE} and~\ref{prop:3ODE}.
In particular, 
    if $k(\alpha(\ell(x)-1),\tu y)=0$, $f(x, \tu y)= f(\alpha(\ell(x)-1), \tu y)$.
    If $k(\alpha(\ell(x)-1),\tu y)=1$,
    \begin{align*}
        f(x,\tu y) &= f(\alpha(\ell(x)-1), \tu y)- n \times \bigg\lfloor \frac{f(\alpha(\ell(x)-1),\tu y)}{n}\bigg\rfloor +1 \\
        &= \begin{cases}
            f(\alpha(\ell(x)-1), \tu y) +1 \quad &\text{if } f(\alpha(\ell(x)-1), \tu y) < n - 1 \\
            f(\alpha(\ell(x)-1), \tu y) - n + 1 \quad &\text{if } f(\alpha(\ell(x)-1), \tu y) = n-1
        \end{cases} \\
        &= \begin{cases}
            f(\alpha(\ell(x)-1), \tu y) +1 \quad \quad \; \; \; &\text{if } f(\alpha(\ell(x)-1), \tu y) < n - 1 \\
            0 \quad \quad \; \; &\text{if } f(\alpha(\ell(x)-1), \tu y) = n-1
        \end{cases}
    \end{align*}
    Clearly, for any $x$ and $\tu y$, $f(x,\tu y)\leq n-1$.
    Additionally, when $k(\alpha(\ell(x)-1), \tu y)=1$, $f(x,\tu y)= (f(\alpha(\ell(x)-1), \tu y)+1$) mod $n$.
    Hence,
    $
    f(x,\tu y) = \sharp \{t \in \{-1, \dots, \ell(x)-1\} : k(\alpha(t), \tu y) \neq 0\}$ \textsc{ mod } n, where $k(\alpha(-1),\tu y)=g(\tu y)$.
    The corresponding computation is doable using constant-depth circuits including \textsc{Mod $n$} gates:
    \begin{itemize}
        \itemsep0em
        \item In the beginning compute $g(\tu y)$ and, for any $t\in \{0,\dots, \ell(x)-1\}$, $k(\alpha(t),\tu y)$'s independently and in parallel.
        This can be done in $\ACCp{n}$ by hypothesis.
        \item In one step, compute their sum modulo $n$, by using a single \textsc{Mod $n$} gate.
    \end{itemize}
\end{proof}

The following proof is simply a special case of the one below, showing that a alternative version of the proof of Prop.~\ref{prop:2ODE} can be provided.

\begin{proof}[Alternative proof of Prop.~\ref{prop:2ODE}]
    If $k(\alpha(\ell(x)-1),\tu y)=0$, $f(x, \tu y)= f(\alpha(\ell(x)-1), \tu y)$.
    If $k(\alpha(\ell(x)-1), \tu y)=1$, there are two possible cases:
    \begin{align*}
        f(x,\tu y) &= f(\alpha(\ell(x)-1), \tu y) - 2f(\alpha(\ell(x)-1),\tu y) +1 \\
        &= \begin{cases}
            1 \quad &\text{if } f(\alpha(\ell(x)-1), \tu y)= 0 \\
            0 \quad &\text{if } f(\alpha(\ell(x)-1), \tu y) = 1.
        \end{cases}
    \end{align*}
    Clearly, for any $x$ and $\tu y$, $f(x,\tu y)\leq 1$.
    Since for $k(t,\tu y)=0$, we can step $t$, $f(x,\tu y)=\sharp\{t \in \{-1, \dots, \ell(x)-1\} : k(\alpha(t),\tu y)\neq 0\}$ mod 2, where $k(\alpha(-1), \tu y)=g(\tu y)$.
    The corresponding computation is feasible in constant-depth for polynomial-sized unbounded fan-in circuits with \textsc{Mod 2} gates:
    \begin{itemize}
        \itemsep0em
        \item In the beginning compute $g(\tu y)$ and, for any $t\in \{0,\dots, \ell(x)-1\}$, $k(\alpha(t),\tu y)$'s independently and in parallel.
        This can be done in $\ACCt$ by hypothesis.
        \item In one step, compute their sum modulo 2, by using a single \textsc{Mod 2} gate.
    \end{itemize}
\end{proof}
\end{toappendix}

\begin{theorem}\label{th:FACCn}
    For any $n\in \Nat^{>1}$, $\ACCp{n}=[\mathscr{B}_0; \circ, \ell\emph{-ODE}_1, \ell\text{-}n\emph{ODE}]$.
\end{theorem}

\begin{proof}[Proof Sketch]
$(\supseteq)$ Following~\cite{ADK24a}, plus Prop.~\ref{prop:nODE} for $\ell$-$n$ODE.
$(\subseteq)$ Generalizing results for constant-depth circuits computation (via $\ACDL$) with \textsc{Mod} $n$ gates, which can be rewritten in our class as functions to count modulo $n$, namely,~$\fun{cmodn}(x,x)$, where $\fun{cmodn}$ is defined as an instance of $\ell$-$n$ODE being the solution of the IVP s.t.~$g(\tu y)=0$ and $k(x,\tu y)=\fun{BIT}(\ell(x),\tu y)$.
\end{proof}

\begin{toappendix}
\begin{proof}[Proof of Theorem~\ref{th:FACCn}]
    $(\supseteq)$ For functions and schemas in $\ACDL$ the proof is as in~\cite{ADK24a}, while for $\ell$-$n$ODE it is due to Prop.~\ref{prop:nODE}.
    $(\subseteq)$ Computation through constant-depth, polynomial-size circuits is captured by $\ACDL$, while counting using \textsc{Mod $n$} gates can be obtained by a function $\fun{cmod n}(x,x)$ defined by $\ell$-$n$ODE from $g$ and $k$, such that $g(\tu y)=0$ and $k(x,\tu y)=\fun{BIT}(\ell(x), \tu y)$, i.e.,~as the solution of the IVP with initial value $\fun{cmodn}(0,\tu y)=0$ and such that $\frac{\partial \fun{cmod n}(x,\tu y)}{\partial \ell}=- n\times \bigg\lfloor \frac{\fun{cmod n}(x,\tu y)}{n-1} \bigg\rfloor \times \fun{BIT}(\ell(x), \tu y) + \fun{BIT}(\ell(x), \tu y)$ (see App.~\ref{app:ACCp}).
\end{proof}
\end{toappendix}

\begin{remark}\label{remark:ntODE}
Clearly, $\ell$-2ODE is nothing but a special case of $\ell$-$n$ODE for $n=2$ so that Prop.~\ref{prop:2ODE} can be alternatively established following the intuition of Th.~\ref{th:FACCn}. 
\end{remark}

\begin{toappendix}
\begin{proof}[Proof of Remark~\ref{remark:ntODE}]
The $\ell$-2ODE is a special case of $\ell$-$n$ODE such that $n=2$,
\begin{align*}
\frac{\partial f(x,\tu y)}{\partial \ell} &= -2 \times k(x,\tu y) \times \lfloor f(x,\tu y)\rfloor + k(x,\tu y) \\
&= -2 \times k(x,\tu y) \times f(x, \tu y) + k(x,\tu y).
\end{align*}
Indeed, for any $x$ and $\tu y$, $f(x,\tu y)$ takes only values 0 or 1.
\end{proof}

\end{toappendix}

\noindent
As expected, characterizations in the non-uniform setting are established without substantial modifications with respect to Prop.~\ref{prop:nonUnif}.

Before concluding Sec.~\ref{sec:ACCn}, we remark that an equivalent characterization for these classes can be obtained not only by considering an alternative family of strict schemas s.t.~$\frac{\partial f(x,\tu y)}{\partial \ell} = -n \times \Big\lfloor \frac{f(x,\tu y)+k(x,\tu y)}{n}\Big\rfloor + k(x,\tu y)$, but also their \emph{non-strict} counterpart, obtained by allowing only special calls to $f$.

\begin{remark}[Non-strict counterpart]\label{remark:nonstrict}
    Let $n\in \Nat^{>2}$, $g:\Nat^p\to \{0,\dots, n-1\}$ and $k: \Nat^{p+1}\to \Nat$. The function $f:\Nat^{p+1}\to \Nat$ defined as the solution of the IVP s.t.~$f(0,\tu y)=g(\tu y)$ and $
    \frac{\partial f(x,\tu y)}{\partial \ell} = -n \times \fun{sg}(f(x,\tu y)-c) \times k(x,\tu y) + k(x,\tu y)
    $
    for $c=n-2$ is exactly the function computed by $\ell$-$n$ODE from the same $g$ and $k$.
\end{remark}

\noindent
The systematic investigation of the relationship between strict and non-strict schemas, which has just been initiated here, deserves further attention but is left for future work.

\begin{toappendix}
\subsubsection{An alternative approach}\label{app:additional}
\paragraph{On $\ell$-$n$ODE*}
\begin{definition}[$\ell$-$n$ODE$^*$schema]
\begin{sloppypar}
    Let $n\in \Nat^+$, $g:\Nat^p \to \{0,\dots, n-1\}$ and $k : \Nat^{p+1}\to \{0,1\}$.
    Then, $f:\Nat^{p+1} \to \Nat$ is defined by $\ell$\emph{-$n$ODE$^*$} from $g$ and $k$ if it is the solution of the IVP with initial value $f(0,\tu y)=g(\tu y)$ and such that:
    $$
    \frac{\partial f(x,\tu y)}{\partial \ell} = - n \times \bigg\lfloor \frac{f(x,\tu y)+k(x,\tu y)}{n}\bigg\rfloor + k(x,\tu y).
    $$
\end{sloppypar}
\end{definition}
\noindent
Notice that by construction, for any $\tu y$, $g(\tu y) \leq n-1$ so that $f(x,\tu y)$ can only take values in $\{0,\dots, n-1\}$.
It is easy to see that, for any $t\in \{0,\dots, x\}$ and $t' \in \{0,\dots, x\}$ such that $\ell(t')=\ell(t)+1$,  if $k(t,\tu y)=0$, then $f(t', \tu y)=f(t,\tu y)$; if $k(t,\tu y)=1$, then:
\small
\begin{align*}
    f(t',\tu y) &= f(t, \tu y) - n \times \bigg\lfloor \frac{f(t,\tu y)+1}{n}\bigg\rfloor + 1 \\
    &= \begin{cases}
        f(t, \tu y) + 1 \; &\text{if } f(t, \tu y) < n-1 \\
        (n-1) - n \times \bigg\lfloor \frac{(n-1)+1}{n}\bigg\rfloor + 1 =0  \; &\text{if } f(t, \tu y) = n-1 
    \end{cases} \\
    %
\end{align*}
\normalsize
\noindent
While the limit case $n=1$ is not very interesting, for 
any $m\in \Nat^{>1}$, computation performed by $\ell$-$m$ODE$^*$ is precisely equivalent to that performed by corresponding $\ell$-$m$ODE.
Accordingly, the characterization established in Th.~\ref{th:FACCn} naturally holds in this setting.
\begin{corollary}\label{cor:nODE}
    For any $n\in \Nat^{>1}$, $\ACCp{n}=[\mathscr{B}_0; \circ, \ell$\emph{-ODE$_1$}, $\ell$\emph{-$n$ODE}*$]$
\end{corollary}
\paragraph{On non-strict schemas for counting modulo $n$}
\begin{proof}[Proof of Remark~\ref{remark:nonstrict}]
    For any $n\in \Nat^{>1}$, it is clear that the initial value is the same.
    Then, considering computation steps, for any $t \in \{0,\dots,x\}$ and $t'\in \{0,\dots, x\}$ such that $\ell(t')=\ell(t)+1$, since $f(t, \tu y)< n-1$: if $k(t,\tu y)=0$, $f(t',\tu y)=f(t,\tu y)$; if $k(t,\tu y)=1$,
    \begin{align*}
        f(t', \tu y) &= - n \times \fun{sg}\big(f(t,\tu y) - c\big) + 1 \\
        &= \begin{cases}
             f(t, \tu y) + 1 \quad \; \; \;   &\text{if } f(t,\tu y) \leq c \\
             f(t, \tu y) - n + 1 &\text{if } f(t, \tu y) = c + 1
        \end{cases} \\
        &= \begin{cases}
             f(t, \tu y) + 1 \quad \quad &\text{if } f(t,\tu y) < n - 1 \\
             0 &\text{if } f(t, \tu y) = n - 1 
        \end{cases}
    \end{align*}
    which is precisely the computation performed by $\ell$-$n$ODE* (recall that $c=n-2$).
\end{proof}
\begin{corollary}
    For any $n\in \Nat^{>1}$, the algebra defined as the closure of $\mathscr{B}_0$ under composition, \emph{$\ell$-ODE$_1$} and the ($n$ instance of the) schema of Remark~\ref{remark:nonstrict} precisely characterizes $\ACCp{n}$. 
\end{corollary}
\noindent
In fact, this proof holds even for the stronger version of the schema allowing for which $g:\Nat^{p} \to \Nat$.
\end{toappendix}

\subsection{Reaching $\TC^0$, again}\label{sec:TC}

The family of schemas in Def.~\ref{def:nODE} $\ell$-$n$ODE have been so defined that $n \in \Nat^{>2}$. 
Indeed, a schema analogous to $\ell$-$n$ODE but with $n=0$ is stronger than any instance of $\ell$-$n$ODE as, together with $\ACDL$, it is enough to capture $\TC^0$.

\begin{definition}[$\ell$-0ODE schema]
    Let $g:\Nat^p \to \{0,1\}$ and $k:\Nat^{p+1} \to \{0,1\}$, then $f:\Nat^{p+1}\to \Nat$ is defined by $\ell$\emph{-0ODE} from $g$ and $k$ if it is the solution of the IVP with initial value $f(0,\tu y)=g(\tu y)$ and s.t.~$
    \frac{\partial f(x,\tu y)}{\partial \ell} =  k(x, \tu y).
    $
\end{definition}

\noindent
Computation performed by this schema can be easily obtained by iterated (bit-)sum of the values of $k$, 
which is in $\TC^0$~\cite{Vollmer}, so that the desired closure property follows. 
Conversely, it is easy to see that $\ell$-0ODE can compute \textsc{BCount}, considering $k(x,\tu y)=\fun{BIT}(\ell(x), \tu y)$ and computing $f(x,x)$. The result below follows immediately (see App.~\ref{app:TC}).

\begin{theorem}\label{th:TC}
    $\TC^0 = [\mathscr{B}_0; \circ, \ell\emph{-ODE}_1, \ell\text{-}0\emph{ODE}]$
\end{theorem}

This theorem, together with the fact that $\ell$-0ODE can simulate not only bit-counting but also counting modulo $n$ for any $n\in \Nat$, gives us an extremely simple proof that $\ACCp{n}$ is included in $\TC^0$ in a totally machine-independent framework.

\begin{proposition}\label{prop:nODEzODE}
    For any $n\in \Nat^{>1}$, 
    $
    [\mathscr{B}_0; \circ, \ell\emph{-ODE}_1, \ell\emph{-}n\emph{ODE}] \subseteq [\mathscr{B}_0; \circ, \ell\emph{-ODE}_1, \ell\emph{-}0\emph{ODE}].
    $
\end{proposition}
\begin{proof}[Proof Sketch]
    %
    Let $\fun{bcount}(x,y)$ be defined by $\ell$-0ODE from $g$ and $k$ s.t.~$g(\tu y)=0$ and $k(x,\tu y)=\fun{BIT}(x,y)$.
    Clearly, $\fun{bc}(x) = \fun{bcount}(x,x)$ counts the number of 1's in the binary representation of $x$.
    We conclude by defining $\fun{mod_n}(x)=  \fun{bc}(x) - \Big\lfloor \frac{\fun{bc}(x)}{n} \Big\rfloor \times n$, which clearly returns the sum modulo $n$ of the 1's in the binary representation of the $x$, and whose construction only involves functions and schemas in the algebra of Th.~\ref{th:TC}. 
\end{proof}

\begin{corollary}\label{cor:TC}
    For any $n\in \Nat$, $\ACCp{n}\subseteq \mathbf{FACC^0} \subseteq \TC^0$.
\end{corollary}

\begin{toappendix}
\subsection{Proofs from Section~\ref{sec:TC}}\label{app:TC}



\begin{proposition}\label{prop:0ODE}
If $f$ is defined by $\ell$\emph{-0ODE} from $g$ and $k$ computable in $\TC^0$, then $f$ is also computable in $\TC^0$.
\end{proposition}
\begin{proof}
    By Df.~\ref{def:ODE}, $f(x,\tu y)=\sum^{\ell(x)-1}_{u=-1} k(\alpha(u),\tu y)$, with the convention that $k(\alpha(-1), \tu y)=g(\tu y)$.
    This is nothing but an iterated addition, which is computable in $\TC^0$~\cite{Vollmer}.
\end{proof}

\begin{remark}\label{remark:ODEzgen}
    Actually closure of $\ell$-0ODE can be generalized so that $k:\Nat^{k+1}\to \Nat$ and the closure property is established in the same way. 
\end{remark}

\begin{proof}
    The proof is precisely as for Prop.~\ref{prop:0ODE}, but here really exploiting the power of iterated addition (and not simply iterated bit-addition).
\end{proof}

\begin{proof}[Proof of Th.~\ref{th:TC}]
   $(\supseteq)$ Basic functions and schemas other than $\ell$-0ODE are already in $\ACDL$, plus the closure property of $\ell$-0ODE (Prop.~\ref{prop:0ODE}).
    \\
    $(\subseteq)$ Due to the equivalence between this ODE-based function algebra and $\mathbb{TCDL}$ (characterizing $\TC^0$)~\cite{ADK24a}.
\end{proof}

\begin{proof}[Proof of Prop.~\ref{prop:nODEzODE}]
Let $\fun{bcount}(x, y)$ be defined as the solution of the IVP below:
\begin{align*}
    f(0, y) &= 0 \\
    \frac{\partial f(x,y)}{\partial \ell} &= \fun{BIT}(x,y).
\end{align*}
Clearly, this is an instance of $\ell$-0ODE and, since $\fun{BIT}$ is in $\mathscr{B}_0$, it is definable in $[\mathscr{B}_0; \circ, \ell$-ODE$_1$, $\ell$-0ODE$]$.
As, by Df.~\ref{def:ODE}, the solution of this system is $f(x,\tu y) = \sum^{\ell(x)}_{u=0}\fun{BIT}(x,y)$, when computing $\fun{bc}(x)=\fun{bcount}(x,x)$ we are precisely computing \textsc{Bcount}.
Then, the function $\fun{mod_n}$ computing the sum modulo $n$ of the bits in the binary representation of its input is obtained by simply considering computing $\fun{bc}(x)-\big\lfloor \frac{\fun{bc}(x)}{n}\big\rfloor \times n$.
All operations involved are in [$\mathscr{B}_0; \circ, \ell$-ODE$_1$, $\ell$-0ODE$]$ as, in particular, multiplication and integer division are known to be in $\TC^0$~\cite{Hesse} and the given algebra has been proved complete with respect to this class (Th.~\ref{th:TC}).
\end{proof}

\noindent
Due to Remark~\ref{remark:ODEzgen}, the same characterization can be obtained considering the closure of $\mathscr{B}_0$ under $\circ$, $\ell$-ODE$_1$ and \emph{generalized} $\ell$-0ODE.
\end{toappendix}


\section{From ODE-based algebras to bounded arithmetic}\label{sec:bArit}

\begin{toappendix}
\section{Proofs from Section~\ref{sec:bArit}}
In~\cite{CookNguyen}, together with a second-order theory to characterize $\TC^0$, Cook and Nguyen defined a family of theories, called $\bth{V^0(p)}$ to capture $\ACCp{p}$.
In~\cite{Johannsen}, Johannsen introduced a first-order theory, called $\overline{\bth{R_0}}$, to characterize $\TC^0$, showing that the functions which are $\Sigma^b_1$-definable in $\overline{\bth{R_0}}$ are precisely those in $\TC^0$.
Later, another minimal theory for this class, called $\Delta_1^b$-$\bth{CR}$, was introduced in~\cite{JP98}.
In the same period, in~\cite{Johannsen}, the equational theories $\bth{A2V}$ and $\bth{TV}$ were defined to capture $\ACCt$ and $\TC^0$, resp. These theories are isomorphic with Cook and Nguyen's $\bth{V^0}(2)$ and $\bth{VTC^0}$.
\end{toappendix}

Inspired by the ODE-design of schemas presented in the previous section, we introduce proof-theoretical characterizations for constant-depth classes, including modulo counting.
In Sec.~\ref{sec:TAC}, we recall the definition of the first-order theory $\bth{TAC^0}$, introduced in~\cite{CloteTakeuti} to capture $\AC^0$, and show that passing through $\ACDL$ simplifies the 
proof. 
In Sec.~\ref{sec:BA}, we present new purely syntactical characterizations for all classes $\ACCp{n}$, by introducing a family of simple ``counting-modulo $n$'' axioms.\arxiv{(parametrized on the corresponding $n$).} Finally, in Sec.~\ref{sec:BA characterization with new rule}, we propose an alternative approach, this time considering a new rule.

\subsection{Bounded theory for $\AC^0$, revised}\label{sec:TAC}

In~\cite{CloteTakeuti}, it is shown that the functions in $\AC^0$ are precisely those which are esb-definable in the first-order theory $\bth{TAC^0}$ (Def.~\ref{def:TAC}).
While the original proof passes through Clote's $\mathcal{A}_0$~\cite{Clote1990}, here we consider $\ACDL$ as our auxiliary algebra, yielding a more direct characterization proof.
The main difference lies in the treatment of cases involving Bit Comprehension, which intuitively corresponds to the computation performed by $\ell$-ODE$_1$, rather than by CRN.


%
A formula is said to be \emph{essentially sharply bounded in a theory $T$}, if it is in the smallest family including a set of atomic formulas, closed under Boolean connectives and sharply bounded quantification and s.t., if $A(\tu a, x)$ and $B(\tu a, x)$ are in the class and  $T \vdash \exists !x \leq s(\tu a).A(\tu a,x)$, then also $\exists x \leq s(\tu a) (A(\tu a, x) \wedge B(\tu a, x))$ and $\forall x\leq s(\tu a) (A(\tu a, x) \to B(\tu a, x))$ are.

\begin{definition}[esb-definability]
    A function is said to be \emph{esb-definable} in a theory $T$ if there is an esb-formula $A$ and a term $t$ s.t. 
        $(i)$ $T\vdash \forall \tu x \exists y \leq t(\tu x).A(\tu x, y)$
        $(ii)$ $T \vdash \forall \tu x \forall y \forall z(A(\tu x, y) \wedge A(\tu x, z) \to y = z)$
        and $(iii)$ $\forall \tu x A(\tu x, f(\tu x))$ is satisfied.
\end{definition}

\noindent
\begin{sloppypar}
\noindent
This is a restricted form of $\Sigma^b_k$-definability~\cite{Buss}, where $\Sigma^b_k$-formulas are replaced by esb ones.
\arxiv{
Notably, this is very close to the standard notion of $\Sigma_k^b$-definability~\cite{Buss}, where $\Sigma^b_k$-formulas are replaced by essentially sharply bounded ones.}
\end{sloppypar}
The language of $\bth{TAC^0}$ is made of $\bt{0}$, $\bt{1}$, $|\cdot|$, $\div 2$, +, $-$, \#, $\leq, =, \bt{pad}, \bt{msp}$, where (informally) $|x|$ is \arxiv{corresponds to }the length symbol, $\bt{pad}(x,y)$ to $2^{\ell(y)}\times x$ and $\bt{msp}(x,y)$ to $\big\lfloor \frac{x}{2^y}\big\rfloor$.

\begin{notation}
    \arxiv{In the following, w}
    We will use the following standard abbreviations for logical and arithmetic symbols: for any formula $A$ and $B$, $A\leftrightarrow B$ is a shorthand for $(A\to B) \wedge (B\to A)$, $a \neq b$ for $\neg(a=b)$, $2\times a = a + a$,
    $\bt{mod2}(a)=a-2\times (a\div 2)$ and $\bt{bit}(i,a)=\bt{mod2}(\bt{msp}(a,i))$.
\end{notation}


\begin{definition}[Theory $\bth{TAC^0}$]\label{def:TAC}
    The theory $\bth{TAC^0}$ is defined by axioms for all basic functions and predicates in its language, plus the following peculiar ones:
    \begin{itemize}
        \itemsep0em
        \begin{sloppypar}
        \item Bit Extensionality axiom \emph{(BE):} $|a|=|b|, \forall i < |a| \big(\bt{bit}(i,a) = \bt{bit}(i,b)\big) \to a=b$
        \end{sloppypar}
        \item Bit Comprehension axiom scheme \emph{(BC):} $\exists y < 2^{|s(\tu a)|}\forall i < |s(\tu a)|\big(\bt{bit}(i, y) = 1 \leftrightarrow A(\tu a, i)\big)$, where $A(\tu a, i)$ is esb
    \end{itemize}
    and the inference rule \emph{esb-LIND}:
    \begin{prooftree}
        \AxiomC{$A(a), \Gamma \Rightarrow \Delta, A(a+1)$}
        \RightLabel{\emph{esb-LIND}}
        \UnaryInfC{$A(0), \Gamma \Rightarrow \Delta, A(|t|)$}
    \end{prooftree}
    where $A(a)$ is esb and $a$ satisfies the eigenvariable condition.
\end{definition}

The characterization proof for $\AC^0$ in terms of esb-definability in $\bth{TAC}^0$ is structurally similar to that in~\cite{CloteTakeuti} but, as anticipated, it relies on $\ACDL$ rather than Clote's $\mathcal{A}_0$~\cite{Clote1990}.

\begin{theorem}\label{th:TAC}
    A function $f$ is in $\AC^0$ iff it is esb-definable in $\bth{TAC}^0$
\end{theorem}
\begin{proof}[Proof Sketch]
    $(\subseteq)$ By induction on the structure of functions in $\ACDL$ (via~\cite{ADK24a}). All base and inductive cases except for $\ell$-ODE$_1$ are standard. Let $f$ be defined by $\ell$-ODE$_1$ from $g$ and $k$ in $\ACDL$ with codomain $\{0,1\}$.
    W.l.o.g. let $g(\tu y)=0$.
    By Df.~\ref{def:ODE}, $f(x,\tu y)=\sum^{\ell(x)-1}_{i=0} 2^{\ell(x)-i} \times k(x,\tu y)$, that is~$\fun{BIT}(j, f(x,\tu y))= k(\alpha(\ell(x)- (j+1)), \tu y)$, for any $j\in\{0,\dots, \ell(x)-1\}$.
    By IH, $k$ is esb-definable and, by (BC), $f(m,\tu n)$ is uniquely defined by $y$, $\exists y \leq 2^{|m|}\forall i \leq |m|(\bt{bit}(i,y)=1 \leftrightarrow k(2^{|m|-(i+1)}-1, \tu n)=1)$.
    %
    %
    $(\supseteq)$ By simultaneous induction on the free cut free derivation height.
    %
    The most interesting case is that of derivation of height 0 defined by an initial sequent (BC).
    W.l.o.g.~let $A$ be sharply bounded (by IH).
    It can be shown that, for any sharply bounded formula, there is a corresponding function in $\ACDL$ which is 1 when the formula holds and 0 otherwise. 
    Let $k_A$ be the function corresponding to $A(\tu a, i)$.
    Then, the desired function is constructed by defining $f_a$ as the solution of the IVP s.t.~$f_a(0,\tu y)=0$ and $\frac{\partial f_a(x)}{\partial \ell}=f_a(x)+k_A(\ell(x),\tu a)$. 
    We conclude by considering $f_A(\tu a)=f_a(2^{\ell(
    s(\tu a))}, \tu a)$, in $\ACDL$, s.t.~$f_A(\tu a) \leq 2^{|s(\tu a)|}$ and $\forall i \leq |s(\tu a)|\big(\bt{bit}(i,f_A(\tu a))=1 \leftrightarrow k_A(\ell(x)-(i+1), \tu a)\big)$ are 
    satisfied.
\end{proof}

\begin{toappendix}

\subsection{Proofs from Section~\ref{sec:TAC}}\label{app:TAC}

\begin{definition}\label{def:TACax}
    The theory $\bth{TAC^0}$ is defined by the following basic axioms:
    %

\begin{description}
    \itemsep0em
    \item[\bt{ax1}] $a\leq b \to a \leq b+1$
    \item[\bt{ax2}] $a \neq a+1$
    \item[\bt{ax3}] $a \leq a$
    \item[\bt{ax4}] $a\leq b \wedge a \neq b \leftrightarrow a+1 \leq b$
    \item[\bt{ax5}] $a \neq 0 \to 2\times a \neq 0$
    \item[\bt{ax6}] $a\leq b \vee b \leq a$
    \item[\bt{ax7}] $a\leq b \wedge b\leq a \to a = b$
    \item[\bt{ax8}] $a\leq b \wedge b\leq c \to a\leq c$
    \item[\bt{ax9}] $|\bt{0}| = \bt{0}$
    \item[\bt{ax10}] $|1| = 1$
    \item[\bt{ax11}] $a \neq 0 \to \ell(2\times a) = \ell(a)+1 \wedge \ell(2\times a+1) = \ell(a)+1$ 
    \item[\bt{ax12}] $a\leq b \to \ell(a) \leq \ell(b)$
    \item[\bt{ax13}] $|a\# b|= (|a|\times |b|)+1$
    \item[\bt{ax14}] $0\# a=1$
    \item[\bt{ax15}] $a \neq 0 \to 1\#(2\times a) = 2\times (1\# a) \wedge 1\#(2\times a+1) = 2\times (1\# a)$
    \item[\bt{ax16}] $a\#b=b\#a$,
    \item[\bt{ax17}] $|a|=|b| \to a\#b = b\#c$
    \item[\bt{ax18}] $|a|=|b|+|c| \to a\# d=(b\#d) \times (c\# d)$
    \item[\bt{ax19}] $a\leq a+b$
    \item[\bt{ax20}] $a\leq b \wedge a\neq b\to (2\times a)+1 \leq 2\times b\wedge (2\times a)+1 \neq 2\times b$
    \item[\bt{ax21}] $a+b=b+a$
    \item[\bt{ax22}] $a+0=a$
    \item[\bt{ax23}] $a+(b+1) = (a+b)+1$
    \item[\bt{ax24}] $(a+b)+c = a+(b+c)$
    \item[\bt{ax25}] $a+b \leq a+c \leftrightarrow b\leq c$
    \item[\bt{ax26}] $a\neq 0 \to |a| = |a\div 2|+1$
    \item[\bt{ax27}] $a= b\div 2 \leftrightarrow 2 \times a = b\vee (2\times a)+1 = b$
    \item[\bt{ax28}] $\bt{pad}(a,0)=a$
    \item[\bt{ax29}] $\bt{pad}(a,b)=\bt{pad}(a, b\div 2) + \bt{pad}(a, b\div 2)$ provided $b\neq 0$
    \item[\bt{ax30}] $\bt{msp}(a,0) = a$
    \item[\bt{ax31}] $\bt{msp}(a,i+1)=\bt{msp}(a,i)\div 2$
\end{description}


\normalsize
\noindent
plus the peculiar axioms 
\begin{itemize}
    \itemsep0em
    \item[] Bit Extensionality axiom \emph{(BE)}:
    $$
    |a| = |b|, \forall i<|a| \big(\bt{bit}(i,a) = \bt{bit}(i,b)\big) \to a=b
    $$
    \item[] Bit Comprehension \emph{(BC)} axiom scheme:
    $$
    \exists y < 2^{|s(\tu a)|}\forall i < |s(\tu a)| \big(\bt{bit}(i,y) = 1 \leftrightarrow A(\tu a, i)\big)
    $$
    where $A(\tu a, i)$ is esb
\end{itemize}
and the inference rule \emph{esb-LIND}:

\begin{prooftree}
    \AxiomC{$A(a), \Gamma \Rightarrow \Delta, A(a+1)$}
    \RightLabel{\emph{esb-LIND}}
    \UnaryInfC{$A(0), \Gamma \Rightarrow \Delta, A(|t|)$}
\end{prooftree}
where $A(a)$ is esb and $a$ satisfies the eigenvariable condition.
\end{definition}

Following Buss, we will deal with a classical system $\bth{LK}$ endowed with esb-LIND and axioms from Df.~\ref{def:TACax}, so that the notion of proof is enlarged to allow initial sequents of the form $\Rightarrow A$, where $A$ is any substitution instance of an axiom listed above.
Recall that, for any rule of the sequent calculus $\bth{LK}$ except cut, we call the formula explicitly occurring in its lower sequent as \emph{principal}, formulas explicitly occurring in the upper sequent(s) are called \emph{auxiliary} (except for weakening), and formulas occurring in the cedents are called \emph{side} formulas.
Given the $\bth{LK}$-derivation $\mathscr{P}$ of a sequent $\Gamma \Rightarrow \Delta$, a formula is said to be \emph{anchored} by the sequent if it is a direct descendent of a formula occurring in an initial sequent in $\Gamma \Rightarrow \Delta$.
A cut inference in $\mathscr{P}$ is \emph{anchored} if either (i) the cut formula is not atomic and at least one of the two occurrences of the cut formula in the upper sequents is anchored or (ii) the cut formula is atomic and both of the occurrences of the cut formula in the upper sequents are anchored.
A cut inference which is not anchored is said to be \emph{free}.
A proof is free cut free if there is no free cut occurring in it.
It has been variously shown that for any proof $\mathscr{P}$ of any sequent $\Gamma \Rightarrow \Delta$, there is a free cut free $\mathscr{P}'$ for $\Gamma \Rightarrow \Delta$.
For the detailed presentation of $\bth{LK}$ and related structural properties, see e.g.~\cite{Buss,NegrivonPlato}.

In order to obtain the desired characterization of $\AC^0$ via $\bth{TAC^0}$ some properties are preliminarily established.

\begin{proposition}\label{prop:sbf}
    For any sharply bounded formula $F$, there is a function of $\ACDL$, $f$, such that $F$ holds if and only if $f$ takes value 1.
\end{proposition}
\begin{proof}
The proof is by induction on the structure of $F$. We explicitly consider the two inductive cases corresponding to sharply bounded quantification:
\begin{itemize}
    \item $F=\exists z \leq |t(\tu a)| . \; H(z,\tu a)$.
    Let the desired function $f_A$ be defined as $f_A=\fun{sg}\big(f_e(2^{\ell(t(\tu a))}, \tu a)\big)$ and
    \begin{align*}
        f_e(0, \tu a) &= f_G(\tu a) \\
        \frac{\partial f_e(z,\tu a)}{\partial \ell} &= f_e(z,\tu a) + f_H(\ell(z)+1, \tu a)
    \end{align*}
    being $f_G(\tu a)=f_H(0, \tu a)$ and $f_H(i, \tu a)=1$ when $H(i,\tu a)$ holds and $f_H(i,\tu a)=0$ otherwise; $f_H$ exists and is in $\ACDL$ by IH.
    This is an instance of $\ell$-ODE$_1$, so that, by its closure property, $f_F$ is in $\ACDL$.
    \item $F=\forall z\leq |t(\tu a)|. \; H(z,\tu a)$.
    Similarly, let the desired function $f_F$ be defined as $f_F=\fun{cosg}\big(f_u(2^{\ell(t(\tu a))}, \tu a)\big)$ where $f_u$ is the solution of the IVP below:
    \begin{align*}
        f_u(0,\tu a) &= \fun{cosg}\big(f_G(\tu a)\big) \\
        \frac{\partial f_u(z,\tu a)}{\partial \ell} &= f_u(z,\tu a) + f_H(\ell(z)+1, \tu a)
    \end{align*}
    where, again, $f_G(\tu a)=f_H(0, \tu a)$ and $f_H(i,\tu a)=1$ when $H(i,\tu a)$ holds and $f_H(i,\tu a)=0$ otherwise.
\end{itemize}
\end{proof}

\noindent
The following Proposition allows us to freely use (symbols for) functions which are esb-definable in $\fun{TAC}^0$.

\begin{proposition}[\cite{CloteTakeuti}]
    If $f_1,\dots, f_k$ are esb-definable in $\bth{TAC^0}$ and $T$ is obtained by endowing $\fun{TAC^0}$ with the new symbols $f_1,\dots,f_k$, their defining axioms, the \emph{(BC)} axiom for $A(\tu a, x)$ esb in $T$ and \emph{esb-LIND} for $A(a)$ esb in $T$, then $T$ is a conservative extension of $\bth{TAC^0}$.
\end{proposition}

Theorem~\ref{th:TAC}.$\subseteq$ shows that every function in $\AC^0$ is esb-definable. 
The structure of the proof is similar to that in~\cite{CloteTakeuti}, but, as anticipated, it refers to $\ACDL$ rather than $\mathcal{A}_0$.

\begin{lemma}\label{lemma:FACTAC}
    Every function in $\AC^0$ is esb-definable in $\bth{TAC^0}$.
\end{lemma}

\begin{proof}
    The proof is by induction on the structure of functions in $\ACDL$:
    \begin{itemize}
        \itemsep0em
        \item[] \emph{Base case.} Each basic function of $\ACDL$ is esb-definable in $\bth{TAC^0}$:
        \begin{itemize}
        \itemsep0em
        \item The basic function $\fun{0}$ is esb-definable in $\bth{TAC^0}$ by the esb formula:
        $
        F_0(x,y) := x=x \wedge y=\bt{0}.
        $
        Existence is proved considering $y=0$.
        By rules of conjunction, also $\bth{TAC^0}\vdash x=x \wedge \bt{0}=\bt{0}$.
        As desired, we conclude that $\bth{TAC^0}\vdash \forall x \exists y\leq \bt{0}(x=x\wedge y=\bt{0})$.
        Uniqueness is proved by assuming $\bth{TAC^0}\vdash x=x\wedge y=\bt{0}$ and $\bth{TAC^0}\vdash x=x \wedge z=\bt{0}$.
        By the conjunction rule, in particular $\bth{TAC^0}\vdash z=\bt{0}$ and $\bth{TAC^0}\vdash y=\bt{0}$.
        Then, by the transitivity of identity, we conclude that $\bth{TAC^0}\vdash y=z$, as desired.
        Clearly, in the standard model, if $y=\fun{0}$, for any $n\in \Nat$, $\llbracket n=n \wedge \bt{0}=\overline{y}\rrbracket$ = $\llbracket n=n\rrbracket$ $\&$ $\llbracket 0=0\rrbracket$ = $\top$ $\&$ $\top$.
        \item The basic function 1 is esb-definable in $\bth{TAC^0}$ by the esb formula: $F_1(x,y):= x=x \wedge y=1$.
        The proof is analogous to the one above.
        \item The basic sign function $\fun{sg}$ is esb-definable in $\bth{TAC^0}$ by the esb-formula: $F_{\fun{sg}}(x,y):=(x=0 \wedge y=0) \wedge (x\neq 0 \wedge y=1)$.
        \item The basic length function $\ell$ is esb-definable in $\fun{TAC^0}$ by the esb-formula:
        $F_\ell(x,y) := y = |x|$.
        \item The basic bit function $\fun{BIT}$ is esb-definable in $\fun{TAC^0}$ by the esb-formula: $F_{\fun{BIT}}(x,i,y) := \bt{bit}(i,a)=y$ (or, equivalently, $F_{\fun{BIT}}(x,i,y) := \bt{msp}(x,i) - 2 \times \big\lfloor\frac{\bt{msp}(x,i)}{2} \big\rfloor = y$).
        \item The basic function $\div 2$ is esb-defined in $\bth{TAC^0}$ by the esb-formula:
        $F_{\div 2} (x,y) := x\div 2 = y$
        \item The basic arithmetic functions $\star \in \{+,-,\#\}$ are esb-defined in $\bth{TAC^0}$ by the corresponding esb formula: $F_{\star}(x,x',y):= x\star x'=y$.
        \begin{sloppypar}
        \item The basic function projection $\pi^n_i$ is esb-defined in $\bth{TAC^0}$ by the esb-formula: $F_{\pi^n_i}(x_1,\dots, x_n, y) := \bigwedge_{y\in \{1,\dots, n\}\setminus \{i\}} (x_j=x_j) \wedge (x_i=y)$.
        \end{sloppypar}
        \end{itemize}
    \item \emph{Inductive case.} Functions of $\ACDL$ can be defined either by composition or by $\ell$-ODE$_1$:
    \begin{itemize}
        \itemsep0em
        \item Let $f$ be defined by composition.
        Without loss of generality, we consider the simple case of $f$ being defied by composing the two functions of $\ACDL$, $g$ and $h$.
        For readability, we mostly omit the fixed input $\tu y$.
        Let $h(\tu b)$ be esb-defined by $s(\tu b), H(\tu b,a)$ and $g(a)$ be esb-defined by $t(a), G(a,c)$.
        We conclude that $g\circ h$ is esb-defined by $t(s(\tu b))$ and the formula $\exists x\leq s(\tu b)\big(H(\tu b, x) \wedge G(x, c)\big)$.
        \item Let $f$ be defined by $\ell$-ODE$_1$ from functions of $\ACDL$, $g$ and $k$, taking values in $\{0,1\}$.
        Without loss of generality, let $g(\tu y)=0$ and $f(x,\tu y)$ be the solution of the IVP below:
        \begin{align*}
            f(0,\tu y) &= 0 \\
            \frac{\partial f(x,\tu y)}{\partial \ell} &= f(x,\tu y) + k(x,\tu y).
        \end{align*}
    Therefore, by definition of $\ell$-ODE (Df.~\ref{def:ODE}),
    \begin{align}
     f(x,\tu y)=\sum^{\ell(x)-1}_{i=0}k(\alpha(i), \tu y) \times 2^{\ell(x)-i};   
    \end{align} 
    that is, $\fun{BIT}(0, f(x,\tu y))=k(\alpha(\ell(x)-1),\tu y)$,  $\dots$, $\fun{BIT}(\ell(x)-1, f(x,\tu y))=k(0, \tu y)$.
    In other words, for any $m\in\Nat$ and $i\in \ell(n)$, the following holds:
    $$
        \fun{BIT}(i,f(m,\tu y))=1 \quad \text{iff} \quad k(\alpha(\ell(m)-i), \tu y)=1.
    $$
    Recall that, by IH, $k$ is esb-definable.
    Therefore, by (BC), $f(m, \tu n)$ is uniquely defined by $y$,
    $$
    \exists y \leq 2^{|m|} \forall i \leq |m| \big(\bt{bit}(i, y) = 1 \leftrightarrow k(2^{|m|-(i+1)}-1), \tu n) = 1\big)
    $$
    \end{itemize}
    \end{itemize}
\end{proof}

    %
    %
    %
    %
    %
    %

The other direction Theorem~\ref{th:TAC}.$\supseteq$ is proved relying on Buss' witnessing function method.
In particular, since the rules that define $\bth{TAC^0}$ are all in $\bth{S^1_2}$, the elimination of the free cut is established as in~\cite{Buss}.
The main novelty is the treatment of (BC): passing through $\ACDL$, we have native functions and schemas at our disposal that are especially suitable to deal with sharply bounded formulas (Remark~\ref{remark:minimisation}) and bit encoding of limited (i.e.~logarithmic) length expressions.
This makes the proof much more direct.

\begin{notation}[ep$\Sigma^b_1$-formulas]
    A formula $A$ is said to be \emph{essentially pure $\Sigma^b_1$} (ep$\Sigma^b_1$) if it is of the form $\exists x_1 \leq s_1 \dots \exists x_n \leq s_n A(x_1, \dots, x_n)$, where $A(x_1,\dots, x_n)$ is esb.
    This formula is denoted $A^{sb}$.
    Such definition and notation extend in the predictable ways to sets of formulas, $\Gamma^{sb}, \Delta^{sb},\dots$. 
\end{notation}

\begin{lemma}\label{lemma:TACFAC}
\begin{sloppypar}
    \begin{itemize}
        \itemsep0em
        \item[i.] Every function that is esb-definable in $\fun{TAC^0}$ is in $\AC^0$.
        \item[ii.] If $\Gamma \Rightarrow \Delta$ is provable in $\bth{TAC^0}$, $\Gamma$ and $\Delta$ are made of ep$\Sigma^b_1$-formulas and $\Gamma^{sb}\Rightarrow \Delta^{sb}$ is of the form:
        \begin{align*}
            \exists x_1 \leq s_1(\tu a) A_1(\tu a, x_1), \dots, & \\
            \exists x_m \leq s_m (\tu a) A_m(\tu a, x_m) &\Rightarrow \exists y_1 \leq t_1(\tu a) B_1 (\tu a, y), \\
            & \quad \; \; \dots, \exists y_n \leq t_n(\tu a) B_n(\tu a, y_n)
        \end{align*}
        with $A_i$ and $B_i$ sharply bounded formulas (for any $i\in \{1,\dots, m\})$, then there are functions $f_1, \dots, f_n$ in $\ACDL$ such that:
        \begin{align*}
            b_1 \leq s_1(\tu a), A_1(\tu a, b_1), \dots, & \\
            b_m \leq s_m (\tu a), A_m(\tu a, b_m) &\Rightarrow f_1(\tu a, \tu b) \leq t_1(\tu a) \wedge B_1(\tu a, f_1(\tu a, \tu b)), \\
        &\quad \; \; \dots, f_n(\tu a, \tu b) \leq t_n(\tu a) \wedge B_n(\tu a, f_n(\tu a, \tu b))
        \end{align*}
        is satisfied in the standard model, where esb-definable functions in $\Gamma^{sb}$ and $\Delta^{sb}$ are interpreted as $\AC^0$ functions obtained by (i) and $\tu b=b_1, \dots, b_m$ are all distinct new variables.
    \end{itemize}
\end{sloppypar}
\end{lemma}

\begin{proof}
    The proof is by induction on the derivation height of the free cut free proof from which the esb-definability of the given function $f$ follows (i.e.~namely, deriving the sequent $\Gamma \Rightarrow \Delta$ of ii. above). 
    All cases except for (BC) are standard (see~\cite{CloteTakeuti}).
    If the derivation has height 0 and the sequent is obtained by an instance of the axiom BC, then $\mathscr{P}$ is simply as follows:
    \begin{prooftree}
        \AxiomC{}
        \RightLabel{BC}
        \UnaryInfC{$\Rightarrow \exists y< 2^{|(s(\tu a)|}\forall i < |s(\tu a)| \big(\bt{bit}(i,y)=1 \leftrightarrow A(i,\tu a)\big)$}
    \end{prooftree}
    where $A(i,\tu a)$ is esb.
    For simplicity, let us assume that $F(i,\tu a)$ is sharply bounded (given the inductive construction, we can reduce to considering sharply bounded sub-formulas).
    By Prop.~\ref{prop:sbf}, there is a function of $\ACDL$, $k_A$, such that, for any $i,\tu a$, it returns 1 when $F(i,\tu a)$ holds and 0 otherwise.
    Then, an auxiliary function $f_a$ as the solution of the IVP below:
    \begin{align*}
        f_a(0, \tu a) &= g_A(\tu a) \\
        \frac{\partial f_a(x,\tu a)}{\partial \ell} &= f_a(x,\tu a) + k_A(\ell(x)+1, \tu a)
    \end{align*}
    \begin{sloppypar}
    \noindent
    where $g_A(\tu a)=k_A(0,\tu a)$.
    Since $k_A$ is a function of $\ACDL$ and its output is in $\{0,1\}$ this is an instance of $\ell$-ODE$_1$.
    Thus, $f_a$ is also in $\ACDL$.
    Recall that, by $\ell$-ODE definition, $f_a(x,\tu y)=\sum^{\ell(x)}_{i=0} k_A(i, \tu a) \times 2^{\ell(x)-i}$; that is, for any $i\leq \ell(x)$, $\fun{BIT}(i, f_a(x,\tu a))=k_A(\ell(x)-(i+1),\tu a)$, i.e.,~$\fun{BIT}(i,f(x,\tu a))=1 \leftrightarrow k_A(\ell(x)-(i+1),\tu a)=1$.
    We conclude noticing that  $f_A(\tu a)= f_a(\ell(s(\tu a)), \tu a)$ is the function we are looking for; that is,~$f_a(\tu a) < 2^{|s(\tu a)|}$ and $\forall i<|s(\tu a)| \big(\bt{bit}(i, f_A(\tu a)) = 1 \leftrightarrow k_A(|s(\tu a)|-(i+1), \tu a)=1\big)$.
    \end{sloppypar}
\end{proof}

\end{toappendix}

\subsection{Bounded theories for $\ACCp{n}$ by extending the language}\label{sec:BA}

In this Section, we present a family of homogeneous first-order bounded theories, $\bth{TACn^0}$, capturing all modulo-counting classes $\ACCp{n}$ and the characteristic axioms of which are crucially inspired by the ODE schemas introduced in Sec.~\ref{sec:algebras}. 
The language of $\bth{TACn^0}$ is obtained by endowing that of $\bth{TAC^0}$ by the unary symbol $\bt{cmod-n}$.

\begin{definition}[Theory $\bth{TACn^0}$]\label{df:TACn}
    The theory $\bth{TACn^0}$ is defined by all axioms and schemas of $\bth{TAC^0}$ (Df.~\ref{def:TAC}), plus the following axioms:
    \begin{description}
        \item[$\bt{N1}$]  $\bt{cmod-n}(x)=b$, for $x=b$ and $b\in\{0,1\}$
        \item[$\bt{N2}$] $
        \bt{cmod-n}(x)=\bt{cmod-n}\big(\big\lfloor \frac{x}{2}\big\rfloor\big) - n \times \big\lfloor \frac{\bt{cmod-n}(\lfloor \frac{x}{2}\rfloor)}{n-1}\big\rfloor \times \big(x - 2 \times \big\lfloor \frac{x}{2}\big\rfloor \big) + \big(x - 2 \times \big\lfloor \frac{x}{2}\big\rfloor \big)$. 
    \end{description}
\end{definition}


\noindent
Intuitively, the symbol $\bt{cmod-n}$ is the syntactical counterpart of $\ell$-$n$ODE (or of the function $\fun{cmodn}$; see App.~\ref{app:ACCp}), which counts modulo $n$ \emph{the number of 1's in the binary representation of the input}.
Notice that in the definition above  the symbol for integer division by $n-1$, $\big\lfloor\frac{z}{n-1}\big\rfloor$ is used with a slight abuse of notation, but it is actually a shorthand for $z \leq n-1$.
Similarly, the multiplication symbol is not actually in the language of $\bth{TACn^0}$, but we are here using it to denote bit-multiplication, i.e.,~multiplication of a number by 0 or 1.
%

%


\begin{theorem}\label{th:TACn}
    For $n\in \Nat^{>1}$, a function $f$ is in $\ACCp{n}$ iff it is esb-definable in $\bth{TACn^0}$.
\end{theorem}
\begin{proof}[Proof Sketch]
    $(\subseteq)$ By induction on the structure of functions in the ODE-based algebra presented in Th.~\ref{th:FACCn}.
    Basic functions and inductive steps, except for $\ell$-$n$ODE are treated as in Th.~\ref{th:TAC}.
    Let $f$ be defined by $\ell$-$n$ODE from $g$ and $k$ in $\ACCp{n}$. 
    By Df.~\ref{def:ODE}, this function \arxiv{computes the sum of $g(\tu y), k(\alpha(0), \tu y), \dots, k(\alpha(\ell(x)-1),\tu y)$ modulo $n$.
    This corresponds to }counts modulo $n$ the bit-sum $t(x,\tu y)=\sum^{\ell(x)-1}_{u=-1}k(\alpha(u), \tu y)\times 2^{(\ell(x)-u+1)}$, with the convention that $k(\alpha(-1),\tu y)=g(\tu y)$ (definable by $\ell$-ODE$_1$ from $g$ and $k$, in the given algebra by IH). We conclude by considering 
    the constant bound $n-1$ and the formula $\exists y \leq n-1.\big(\bt{cmod-n}(t(m,\tu n))=y\big)$ esb-defining $f(m,\tu n)$.
    $(\supseteq)$ By induction on the corresponding free cut free proof.
    Since, by Df.~\ref{df:TACn}, only axioms $\bt{N1}$ and $\bt{N2}$ have been added to $\bth{TAC^0}$, we can consider (basic) cases involving \bt{cmod-n} only. Other cases are as in Th.~\ref{th:TAC}.
    In particular, for the only non trivial case $\bt{N2}$, existence and uniqueness easily follow from application of characteristic axioms and (the esb-PIND counterpart of) esb-LIND,
    while for (iii) we consider $\fun{cmn}(x)=\fun{cmodn}(x,x)$, where $\fun{cmodn}$ is defined as in Th.~\ref{th:FACCn}.
    As desired this function is s.t., for any $x\in \Nat$, if $x=0$, $\fun{cmn}(x)=0$, if $x=1$, $\fun{cmn}(x)=1$, and if $x=2z+b$ (with $b\in\{0,1\}$), $\fun{cmn}(x)=\fun{cmn}(z)-n\times (\fun{cmn}(z)\ge n-1)\times \fun{BIT}(x,0)+\fun{BIT}(x,0)$.
    We conclude by noticing that, since $\fun{cmn}$ is in the algebra presented in Th.~\ref{th:FACCn}, by the same theorem, is in $\ACCp{n}$.
\end{proof}


\arxiv{
In this Section, we present bounded theories to characterize corresponding modulo-counting classes.
The characteristic axioms are crucially inspired by the ODE schemas introduced in Sec.~\ref{sec:algebras}.
As for their recursion-theoretic counterparts, with the exception of $\ACCt$ and $\ACCp{6}$, no first-order bounded theories capturing these classes have been introduced before.}

\begin{toappendix}

\subsubsection{From $\bth{TAC2^0}$ to $\bth{TACn^0}$}\label{app:BA}

First, we introduce the theory $\bth{TAC^02}$, a very simple first-order theory the defining axioms of which are designed mirroring the corresponding schema $\ell$-2ODE.
Its definition is then generalized to  the uniform family of minimal theories $\bth{TAC^0n}$, characterizing functions in the corresponding class $\ACCp{n}$.
Differently from the two theories presented in~\cite{CloteTakeuti}, our proposals not only provide characterizations for \emph{all} (constant-depth) circuit classes with counters, but -- having being obtained in an extremely natural way, i.e.~by endowing $\bth{TAC^0}$ with one characterizing axiom -- allows the proofs presented in Sec.~\ref{sec:TAC} to scale to counting classes in a straightforward way.

We start by considering the theory characterizing $\ACCt$.
Its language is obtained by endowing that of $\bth{TAC^0}$ by the symbol $\bt{par}$.

\begin{definition}[Theory $\bth{TAC2^0}$]
    The theory $\bth{TAC2^0}$ includes all axioms and schemas of $\bth{TAC^0}$ (Df.~\ref{def:TAC}), plus the following axioms:
    \begin{description}
    \itemsep0em
    \item[$\bt{P1}$] $\bt{par}(b)=b$, for $b\in\{0,1\}$
    \item[$\bt{P2}$] $\bt{par}(x)=\bt{par}\big(\big\lfloor \frac{x}{2}\big\rfloor\big) - 2 \times \bt{par}\big(\big\lfloor \frac{x}{2}\big\rfloor\big) \times \big(x-2 \times \big\lfloor \frac{x}{2}\big\rfloor\big) + (x-2 \times \big(\big\lfloor \frac{x}{2}\big\rfloor\big)$.
   \end{description}
 \end{definition}

    \arxiv{\begin{itemize}
        \itemsep0em
        \item[] $x=b \to \bt{par}(x)=b$
        \item[] $x=\bt{2}\times y + b\wedge y\neq 0 \to \bt{par}(x)=\bt{par}(y)-2\times \bt{par}(y)\times b + b$,
    \end{itemize}
    for $b\in\{0,1\}$.}

\noindent
These axioms are conceptually simpler than the rule presented in~\cite{CloteTakeuti}.
Moreover, since $\fun{TAC^0}$ is extended by axioms only, the desired characterization proof naturally generalizes the one from Sec.~\ref{sec:TAC}.

\begin{theorem}\label{th:TACt}
    A function $f$ is in $\ACCt$ iff it is esb-definable in $\bth{TAC2^0}$.
\end{theorem}

\begin{proof} 
    $(\subseteq)$ As for Lemma~\ref{lemma:FACTAC}, the proof is by induction on the structure of functions in the ODE-based algebra presented in Th.~\ref{theorem:ACCt}.
    Basic functions and inductive steps, except for $\ell$-2ODE, are proved precisely as before.
    Let us consider the only new case. Let $f(x,\tu y)$ be defined by $\ell$-2ODE from $g$ and $k$. By definition of $\ell$-ODE (Df.~\ref{def:ODE}), it holds that
    $
    f(x,\tu y)= \sum^{\ell(x)-1}_{u=-1} \Big( \prod^{\ell(x)-1}_{t=u+1}\big(1-2k(\alpha(t), \tu y)\big)\Big) \times k(\alpha(u), \tu y),
    $
    with the convention that $k(\alpha(-1),\tu y)=g(\tu y)$.
    Intuitively this function computes the sum of $g(\tu y), k(\alpha(0),\tu y), \dots, k(\alpha(\ell(x)-1), \tu y)$ modulo 2 (Prop.~\ref{prop:2ODE}).
    This is equivalent to computing the bit-parity of $t(x,\tu y) = \sum^{\ell(x)-1}_{u=-1}k(\alpha(u), \tu y) \times 2^{u+1}$, again with the convention that $k(\alpha(-1), \tu y)=g(\tu y)$, which can be defined in the given algebra using $\ell$-ODE$_1$ from $g$ and $k$ (in the algebra by IH and Th.~\ref{theorem:ACCt}).
    We conclude by considering the formula defined considering $\bt{par}(t(x, \tu y))$ and the constant bound 1, i.e.,~$\exists y \leq 1(\bt{par}(t(m,\tu n)) = y)$, esb-defining $f(m,\tu n)$.
\\
    $(\supseteq)$ The proof is again by induction on the corresponding free cut free proof.
    As desired, since only axioms are added to $\bth{TAC^0}$, all cases except for basic ones involving $\bt{par}$ are precisely as in Lemma~\ref{lemma:TACFAC}. The only new (non-trivial)
    case is that of derivations of height 0 of the form:
    \begin{center}
 $
 \Rightarrow \bt{par}(x)=\bt{par}\big(\big\lfloor \frac{x}{2}\big\rfloor\big) - 2 \times \bt{par}\big(\big\lfloor \frac{x}{2}\big\rfloor\big) \times \big(x-2 \times \big\lfloor \frac{x}{2}\big\rfloor\big) + (x-2 \times \big\lfloor \frac{x}{2}\big\rfloor\big)
 $
 \end{center}
We start by proving the existence condition.
First, recall that due to~\cite[Prop. 9.1]{CloteTakeuti}, we can substitute esb-LIND with the equivalent inference rule esb-PIND below:
    \begin{prooftree}
    \AxiomC{$A\big(\big\lfloor\frac{a}{2}\big\rfloor\big),\Gamma \Rightarrow \Delta, A(a)$}
    \RightLabel{esb-PIND}
    \UnaryInfC{$A(0),\Gamma \Rightarrow \Delta, A(t)$}
    \end{prooftree}
 First, notice that, if $x=0$, the desired $y$ is 0, and similarly for $x=1$:
 By $\bt{p1}$, $\bth{TAC2^0} \vdash \bt{par}(0)=0$, so $\bth{TAC2^0} \vdash \exists y\leq 1.\bt{par}(0)=y$ (by rule for $\exists_\leq$), which is clearly an esb-formula. 
 Then, we ``assume'' that $\bth{TAC2^0} \vdash \exists !y \leq 1. \bt{par}\big(\big\lfloor \frac{a}{2}\big\rfloor\big)$, which, due to the defining axioms for $\leq$ ($\bt{0}$ and $\bt{1}$) and rules for bounded existential quantifiers, is equivalent to assuming that 
 $\bth{TAC2^0} \vdash z\leq 1 \wedge \bt{par}\big(\big\lfloor\frac{a}{2}\big\rfloor\big)=z$.
 By the defining axioms for $\leq$ (and $\bt{0}$ and $\bt{1}$), this is equivalent to saying $\bth{TAC2^0} \vdash \bt{par}\big(\big\lfloor \frac{a}{2}\big\rfloor\big) =0 \vee \bt{par}\big(\big\lfloor \frac{a}{2}\big\rfloor \big) = 1$.
 By instantiating \bt{p2}, we have that in particular $\bth{TAC2^0}\vdash \bt{par}(a) = \bt{par}\big(\big\lfloor \frac{a}{2}\big\rfloor\big) - 2 \times \bt{par}\big(\big\lfloor \frac{a}{2}\big\rfloor\big) \times \big(a- 2\times \big\lfloor \frac{a}{2}\big\rfloor\big) + \big(a- 2\times \big(\big\lfloor \frac{a}{2}\big\rfloor\big)\big)$ and, due to equality axioms,
 \begin{itemize}
 \itemsep0em
 \item[a.] either $\bth{TAC2^0} \vdash \bt{par}(a)=0 + \big(a-2 \times \big\lfloor \frac{a}{2}\big\rfloor\big)$
 \item[b.] or $\bth{TAC2^0} \vdash \bt{par}(a)=1-2\times \big(a-2\times \big\lfloor \frac{a}{2}\big\rfloor\big) + \big(a-2\times \big\lfloor \frac{a}{2}\big\rfloor\big)$
 \end{itemize}
 Additionally, by instantiating \bt{ax27}, we know that, in particular, either $\bth{TAC2^0}\vdash a=2\times \big\lfloor \frac{a}{2}\big\rfloor+0$, i.e.~$\bth{TAC2^0}\vdash a=2\times \big\lfloor \frac{a}{2}\big\rfloor$, or $\bth{TAC2^0}\vdash a=2\times \big\lfloor \frac{a}{2}\big\rfloor+1$.
 So,
 \begin{itemize}
 \itemsep0em
 \item[a.] either $\bth{TAC2^0} \vdash \bt{par}(a)=0$ or $\bth{TAC2^0} \vdash \bt{par}(a)=1$
 \item[b.] or, either $\bth{TAC2^0} \vdash \bt{par}(a)= 1$ or $\bth{TAC2^0} \vdash \bt{par}(a)=0$.
 \end{itemize}
 In both cases the desired value exists and is either 0 or 1, so that we can conclude $\bth{TAC2^0} \vdash \exists y \leq 1.\bt{par}(a)=y$.
 Being the formula esb, we derive the desired existence property by esb-PIND.
A bit more formally, existence can be established as follows.
For the sake of readability, applications of exchange rules are not explicitly mentioned.
%
\begin{prooftree}
    \AxiomC{$\mathscr{P}_1$}
    \noLine
    \UnaryInfC{$\bt{par}\big(\big\lfloor\frac{a}{2}\big\rfloor\big) = 0 \Rightarrow \exists x \leq 1. \bt{par}(a)=x$}
    \AxiomC{$\mathscr{P}_2$}
    \noLine
    \UnaryInfC{$\bt{par}\big(\big\lfloor\frac{a}{2}\big\rfloor\big) = 1 \Rightarrow \exists x \leq 1. \bt{par}(a)=x$}
    \RightLabel{$\vee$L}
    \BinaryInfC{$\bt{par}\big(\big\lfloor\frac{a}{2}\big\rfloor\big) = 0 \vee \bt{par}\big(\big\lfloor\frac{a}{2}\big\rfloor\big) = 1 \Rightarrow \exists x \leq 1. \bt{par}(a)=x$}
    \RightLabel{$\exists_\leq$R, $\bt{Ax}_{\leq,0,1}$}
    \UnaryInfC{$\exists x\leq 1. \bt{par}\big(\big\lfloor\frac{a}{2}\big\rfloor\big) = x  \Rightarrow \exists x \leq 1. \bt{par}(a)=x$}
\end{prooftree}

\noindent
where, in particular, $\mathscr{P}_1$ is defined as below and $\mathscr{P}_2$ is defined similarly but assuming $\bt{par}\big(\big\lfloor \frac{a}{2}\big\rfloor\big) = 1$:
\tiny
\begin{prooftree}
\AxiomC{}
\RightLabel{\bt{P2}}
\UnaryInfC{$\Rightarrow \bt{P2}\{a/x\}$}
\RightLabel{\bt{Wk-L}}
\UnaryInfC{$\bt{par}\big(\big\lfloor\frac{a}{2}\big\rfloor\big) = 0, a=2\times \big\lfloor \frac{a}{2}\big\rfloor \Rightarrow \bt{P1}\{a/x\}$}
\RightLabel{\bt{Ax22}, \bt{Eq}}
\UnaryInfC{$\bt{par}\big(\big\lfloor\frac{a}{2}\big\rfloor\big) = 0, a=2\times \big\lfloor \frac{a}{2}\big\rfloor \Rightarrow \bt{par}(a)=0$}
\RightLabel{$\exists_\leq$R}
\UnaryInfC{$\bt{par}\big(\big\lfloor\frac{a}{2}\big\rfloor\big) = 0, a=2\times \big\lfloor \frac{a}{2}\big\rfloor \Rightarrow \exists x \leq 1 .\bt{par}(a)=x$}
\AxiomC{}
\RightLabel{\bt{P2}}
\UnaryInfC{$\Rightarrow \bt{P2}\{a/x\}$}
\RightLabel{\bt{Wk-L}}
\UnaryInfC{$\bt{par}\big(\big\lfloor\frac{a}{2}\big\rfloor\big) = 0, a=2\times \big\lfloor \frac{a}{2}\big\rfloor +1 \Rightarrow \bt{P1}\{a/x\}$}
\RightLabel{\bt{Ax22}, \bt{Eq}}
\UnaryInfC{$\bt{par}\big(\big\lfloor\frac{a}{2}\big\rfloor\big) = 0, a=2\times \big\lfloor \frac{a}{2}\big\rfloor \Rightarrow \bt{par}(a)=1$}
\RightLabel{$\exists_\leq$R}
\UnaryInfC{$\bt{par}\big(\big\lfloor\frac{a}{2}\big\rfloor\big) = 0, a=2\times \big\lfloor \frac{a}{2}\big\rfloor \Rightarrow \exists x \leq 1. \bt{par}(a)=x$}
\RightLabel{$\vee L$}
\BinaryInfC{$\bt{par}\big(\big\lfloor\frac{a}{2}\big\rfloor\big) = 0, a=2\times \big\lfloor \frac{a}{2}\big\rfloor \vee a=2\times \big\lfloor \frac{a}{2}\big\rfloor + 1 \Rightarrow \exists x \leq 1. \bt{par}(a)=x$}
\RightLabel{\bt{Cut}, \bt{Ax27}\{a/x\}}
\UnaryInfC{$\bt{par}\big(\big\lfloor\frac{a}{2}\big\rfloor\big) = 0 \Rightarrow \exists x \leq 1. \bt{par}(a)=x$}
\end{prooftree}
\normalsize
 Uniqueness can be established in a similar way.
    Finally, to ensure the last condition of esb-definability, let us consider $f_{par}(x)=f_p(x,x)$, where $f_p(x,\tu y)$ is defined as the solution of the IVP with initial value $f_p(0,\tu y)=0$ and such that
        $
        \frac{\partial f_p(x,\tu y)}{\partial \ell} = - 2 \times \fun{BIT}(\ell(x), \tu y) \times f_p(x,\tu y) + \fun{BIT}(\ell(x), \tu y).
        $
    As desired, this function is such that, for any $x \in \Nat$, if $x = 0$, $f_{par}(x)=0$, if $x=1$, $f_{par}(x)=1$ and if $x=2z + b$, where $z=  \lfloor \frac{x}{2}\rfloor$ and $b = x- 2\times \lfloor \frac{x}{2}\rfloor \in \{0,1\}$, $f_{par}(x)= f_{par}(z) - 2 \times f_{par}(z) \times \fun{BIT}(z, 0) + \fun{BIT}(z, 0)= f_{par}(z) - 2 f_{par}(z)\times b + b = f_{par}\big(\big\lfloor \frac{x}{2}\big\rfloor\big) - 2 f_{par}\big(\big\lfloor \frac{x}{2}\big\rfloor\big) \times \big(x- 2\times \lfloor \frac{x}{2}\rfloor\big) + \big(x- 2\times \lfloor \frac{x}{2}\rfloor\big)$. 
    We conclude the proof by remarking that, since $f_p$ is defined by $\ell$-2ODE from $\fun{BIT}$, $f_{par}$ is in the function algebra introduced in Sec.~\ref{sec:ACCt}, so by Th.~\ref{theorem:ACCt}, is in $\ACCt$.
\end{proof}

This simple setting can be naturally generalized to obtain homogeneous arithmetic theories for all classes $\ACCp{n}$.
As seen, for any $n\in \Nat^{>2}$,  the language of $\bth{TACn^0}$ is obtained from the language of $\bth{TAC^0}$ by simply adding the corresponding unary symbol $\bt{cmod-n}$.

\begin{proof}[Proof of Theorem~\ref{th:TACn}]
    Existence and uniqueness are proved as in Th.~\ref{th:TACt}. 
    The main difference is that, in this case, $\bt{cmod-n}\big(\big\lfloor\frac{a}{2}\big\rfloor\big)$ can take $n$ different values, namely $\{0,\dots, n-1\}$.
%
    Clearly, by $\bt{N1}$, $\bth{TACn^0} \vdash \bt{cmod-n}(0)=0$.
    %
     In general, based on the value of $a$, there are two main possible cases to be considered, namely $a = 2\times \big\lfloor \frac{a}{2}\big\rfloor +0$ (H1) and $a = 2\times \big\lfloor \frac{a}{2}\big\rfloor+1$ (H2).
    Therefore, for any $m\in \{0, \dots, n-2\}$,
    $$
        \mathscr{P}_m
    $$
    \tiny
\begin{prooftree}
    \AxiomC{}
    \RightLabel{\bt{N2}, \bt{Wk-L}}
    \UnaryInfC{$m\leq n-2, \bt{cmod-n}(\frac{a}{2}) = m, \text{H1} \Rightarrow \bt{N2}\{a/x\}$}
    \RightLabel{\bt{Eq}}
    \UnaryInfC{$m\leq n-2, \bt{cmod-n}(\frac{a}{2}) = m, \text{H1} \Rightarrow \bt{cmod-n}(a) = m$}
    \RightLabel{$\exists_\leq$R}
    \UnaryInfC{$m\leq n-2, \bt{cmod-n}(\frac{a}{2}) = m, \text{H1} \Rightarrow \exists x\leq n-1.\bt{cmod-n}(a) = x$}
    \UnaryInfC{$m\leq n-2, \bt{cmod-n}(\frac{a}{2}) = m, \text{H1} \Rightarrow \exists x\leq n-1.\bt{cmod-n}(a) = x$}
    \AxiomC{}
    \RightLabel{\bt{N2}, \bt{Wk-L}}
    \UnaryInfC{$m\leq n-2, \bt{cmod-n}(\frac{a}{2}) = m, \text{H2} \Rightarrow \bt{N2}\{a/x\}$}
    \RightLabel{\bt{Eq}, \bt{Ax}$_{1,+,\leq}$}
    \UnaryInfC{$m+1\leq n-1, \bt{cmod-n}(\frac{a}{2}) = m, \text{H2} \Rightarrow \bt{cmod-n}(a) = m+1$}
    \RightLabel{$\exists_\leq$R}
    \UnaryInfC{$m+1\leq n-1, \bt{cmod-n}(\frac{a}{2}) = m, \text{H2} \Rightarrow \exists x\leq n-1.\bt{cmod-n}(a) = x$}
    \RightLabel{$\vee$L}
    \BinaryInfC{$m<n-2, \bt{cmod-n}(\frac{a}{2}) = m, \text{H1} \vee \text{H2} \Rightarrow \exists x\leq n-1.\bt{cmod-n}(a) = x$}
    \RightLabel{$\bt{Ax27}$}
    \UnaryInfC{$m<n-2\bt{cmod-n}(\frac{a}{2}) = m \Rightarrow \exists x\leq n-1.\bt{cmod-n}(a) = x$}
\end{prooftree}
\normalsize
and
$$
\mathscr{P}_{n-1}
$$
    \tiny
\begin{prooftree}
    \AxiomC{}
    \RightLabel{\bt{N2}, \bt{Wk-L}}
    \UnaryInfC{$n-1\leq n-1, \bt{cmod-n}(\frac{a}{2}) = n-1, \text{H1} \Rightarrow \bt{N2}\{a/x\}$}
    \RightLabel{\bt{Eq}}
    \UnaryInfC{$n-1\leq n-1, \bt{cmod-n}(\frac{a}{2}) = n-1, \text{H1} \Rightarrow \bt{cmod-n}(a) = n-1$}
    \RightLabel{$\exists_\leq$R}
    \UnaryInfC{$n-1\leq n-1, \bt{cmod-n}(\frac{a}{2}) = n-1, \text{H1} \Rightarrow \exists x\leq n-1.\bt{cmod-n}(a) = x$}
    \RightLabel{\bt{Ax}$_\leq$,\bt{Eq}}
    \UnaryInfC{$\bt{cmod-n}(\frac{a}{2}) = n-1, \text{H1} \Rightarrow \exists x\leq n-1.\bt{cmod-n}(a) = x$}
    \AxiomC{}
    \RightLabel{\bt{N2}, \bt{Wk-L}}
    \UnaryInfC{$0\leq n-1, \bt{cmod-n}(\frac{a}{2}) = n-1, \text{H2} \Rightarrow \bt{N2}\{a/x\}$}
    \RightLabel{\bt{Eq}, \bt{Ax}$_{+,-,1}$}
    \UnaryInfC{$0\leq n-1, \bt{cmod-n}(\frac{a}{2}) = n-1, \text{H2} \Rightarrow \bt{cmod-n}(a) = 0$}
    \RightLabel{$\exists_\leq$R}
    \UnaryInfC{$0\leq n-1, \bt{cmod-n}(\frac{a}{2}) = n-2, \text{H2} \Rightarrow \exists x\leq n-1.\bt{cmod-n}(a) = x$}
    \RightLabel{\bt{Ax}$_\leq$}
    \UnaryInfC{$\bt{cmod-n}(\frac{a}{2}) = n-1, \text{H2} \Rightarrow \exists x\leq n-1.\bt{cmod-n}(a) = x$}
    \RightLabel{$\vee$L}
    \BinaryInfC{$\bt{cmod-n}(\frac{a}{2}) = n-1, \text{H1} \vee \text{H2} \Rightarrow \exists x\leq n-1.\bt{cmod-n}(a) = x$}
\end{prooftree}
\normalsize
where, for space reasons, we used $\frac{a}{2}$ to denote integer division, i.e.~since here it is not ambiguous we used $\frac{a}{2}$ as a shorthand for $\big\lfloor \frac{a}{2}\big\rfloor$.
\normalsize
Similarly to Th.~\ref{th:TACt}, we conclude by deriving and applying esb-LIND: 

\begin{prooftree}
    \AxiomC{$\mathscr{P}_m$'s (for $m\in \{0,\dots, n-2\}$)}
    \AxiomC{$\mathscr{P}_{n-1}$}
    \RightLabel{ps-$\vee$L's}
    \BinaryInfC{$
    \bigvee_{l\in\{0,\dots, n-1\}}
    \bt{cmod-n}\big(\big\lfloor \frac{a}{2}\big\rfloor \big) = l
    \Rightarrow \exists x\leq n-1. \; \bt{cmod-n}(a)=x$}
    \RightLabel{$\exists_\leq$L}
    \UnaryInfC{$\exists x\leq n-1. \; \bt{cmod-n}\big(\big\lfloor \frac{a}{2}\big\rfloor \big) = x \Rightarrow \exists x\leq n-1. \; \bt{cmod-n}(a)=x$}
\end{prooftree}

\end{proof}

\end{toappendix}

\newcommand{\TACzeroBIS}[1]{\bth{TAC}^0[#1]}

\newcommand{\tacp}{\mathsf{TAC^0}[n]}
\newcommand{\tac}{\mathsf{TAC^0}}
\newcommand{\fac}{\mathsf{FAC^0}}
\newcommand{\facp}{\mathsf{FACC}[p]}

\subsection{Bounded theories for  $\ACCp{n}$ by extending the rules}~\label{sec:BA characterization with new rule}~\label{sec:BA characterization with new rule}
In this Section, we propose alternative bounded arithmetic obtained by extending the rule system of the theory without enlarging the language.
The resulting system, $\tacp$, not only captures  $\ACCp{n}$ using the same language as $\bth{TAC}^0$, so to possibly allow a model-theoretic study and comparison between corresponding complexity classes.
For instance, the separation of $\AC^0$ and $\ACCp{n}$  translates to the existence of a model of $\bth{TAC}^0$ that is not a model of $\tacp$. 
In addition, its propositional translation may introduce a new way to formulate propositional proof systems with counting gates. 

Again, the new system is strongly inspired by $\ell$-$n$ODE. However, this approach more faithfully reflects the meaning of ODE in the context of arithmetic. Since ODE schemes state that the totality of difference implies the totality of functions, they are essentially induction schemes in disguise. In particular, the totality of $\fun{cmodn}$ (the function that outputs the number of bits modulo $n$) implies the totality of functions defined by it. 
Informally, it holds
\begin{prooftree}
\AxiomC{$f(x,\tu y)$ is defined from $k$ by $\ell$-$n$ODE}
\AxiomC{$k(x,\tu y)$ is total}
\AxiomC{$f(0,\tu y)$ is defined}
\TrinaryInfC{$f(x,\tu y)$ is total}
\end{prooftree}
which intuitively translates into the following inference rule.

\begin{definition}
Let $n\in \mathbb{N}$ and $B_k$ be some esb-formula. 
Then,
$n$-$\emph{PIND}(B)$ is the following inference rule:
\begin{prooftree}
\AxiomC{$\Gamma,\, b< n \ B(\big\lfloor\frac{a}{2}\big\rfloor, b)\ \Rightarrow\ \Delta,\, \exists z < n \ B (a, z) \wedge A(a,b,z)$}
\AxiomC{$\Gamma\ \Rightarrow\ \Delta, \, \exists y < 2 \  B_k(a, y)$}
\BinaryInfC{$\Gamma, \, b<n \  B (0, b)\ \Rightarrow\ \Delta,\, \exists y\, B (t,y)$}
\end{prooftree}
where $A$ is the formula s.t. 
$A(x,y,z) \iff \exists z'<2\    z = y - n \times z' \times \Big\lfloor  \frac{y}{n-1}\Big\rfloor + z'\wedge B_k(x,z').$
Note that the terms $a, b$ must satisfy the eigenvariable condition in the premise. 
\end{definition}



\begin{definition}
$\tacp$ = $\bth{TAC}^0 + \{ \text{$n$-PIND}(B):B\text{ is esb-formula}\}$. 
\end{definition}

It remains to show that $\tacp$ captures exactly $\ACCp{n}$. 
The soundness proof is relatively standard. The completeness proof is more involved: it requires encoding the history of computation à la Buss. 

\begin{theorem} \label{completeness}
The class of esb-definable formulas in $\tacp$ is exactly $\ACCp{n}$. 
\end{theorem}

\begin{toappendix}

\subsection{Proof of Section~\ref{sec:BA characterization with new rule}}\label{app:ruleBA}

We first recall the new theory $\tacp$ for self-containment. Its language is the same as $\tac$.
However, $\tacp$ extends the set of axioms and rules of $\tac$ by $n$-PIND, which is defined as follows. 
Let $n\in \mathbb{N}$ and $B_k$ be some esb-formula. 
$n$-PIND$(B)$ is the following inference rule: 
\begin{prooftree}
\AxiomC{$\Gamma,\, b< n,\, B(\big\lfloor\frac{a}{2}\big\rfloor, b)\ \Rightarrow\ \Delta,\, \exists z<n\ B (a, z) \wedge A(a,b,z)$}
\AxiomC{$\Gamma\ \Rightarrow\ \Delta, \, \exists y < 2\,  B_k(a, y)$}
\BinaryInfC{$\Gamma, \, b<n, \, B (0, b)\ \Rightarrow\ \Delta,\, \exists y\, B (t,y)$}
\end{prooftree}
where $A$ is the formula s.t. 
$$A(x,y,z) \iff \exists z'<2\   z = y - n \times z' \times \big\lfloor  \frac{y}{n-1}\big\rfloor + z'\wedge B_k(x,z')$$ 
and $a, b$ must satisfy the eigenvariable condition in the premise. 

Now, we intend to simplify the left-hand side of the premises, so that it becomes more readable.

\begin{remark} \label{abuse}
With a slight abuse of notation, we use the following derivation form to denote $n$-PIND($B$): 
\begin{prooftree}
\AxiomC{(\ref{7})}
\AxiomC{(\ref{8})}
\BinaryInfC{(\ref{9})}
\end{prooftree}
where 
\begin{align}
    \Gamma,\, b< n,\, B(\big\lfloor\frac{a}{2}\big\rfloor, b)\ &\Rightarrow\ \Delta,\, B (a, b - n \times k(a) \times \big\lfloor  \frac{b}{n-1}\big\rfloor + k(a))  \label{7} \\
   \Gamma\ &\Rightarrow\ \Delta, \, \exists y < 2\,  B_k(a, y) \label{8} \\
    \Gamma, \, b<n, \, B (0, b)\ &\Rightarrow\ \Delta,\, \exists y\, B (t,y) \label{9}
\end{align}
$B_k$ is the formula defining the function $k$; 
and the use of $k$ directly in the language can be justified by the forthcoming Proposition \ref{conservext}.
\end{remark}

The esb-formulas of $\tacp$ are defined in the same way as $\tac$, and 
$$\tacp = \tac + \{ \text{$n$-PIND}(B):B\text{ is esb-formula}\}$$
As in the case of $\tac$, adding a symbol $f$ for an esb-definable function of $\tacp$ constitutes a conservative extension $\tacp(f)$ of $\tacp$. 

\begin{proposition} \label{conservext}
Let $f$ be an esb-definable function in $\tacp$. Then, for every esb-formula $B$ in $\tacp(f)$, there exists a formula $B'$ without $f$ s.t. $\tacp(f) \vdash B \leftrightarrow B'$.  
\end{proposition}

Clearly, the rule $n$-PIND is designed to capture $\ell$-$n$ODE. This  makes the soundness rather direct to prove. 

\begin{lemma}[Witnessing theorem]~\label{witness}
Let $f$ be an esb-definable function in $\tacp$. Then, $f$ is in $\ACCp{n}$.
\end{lemma}
\begin{proof}
It suffices to extend the proof of Lemma \ref{lemma:TACFAC} to include the rule $n$-PIND. 
We work with the notation of Remark \ref{abuse}. 
By IH, there is a function $f$ in $\ACCp{n}$ s.t. $b< n,\, B(\big\lfloor\frac{a}{2}\big\rfloor, b)\ \Rightarrow\ B (a, f_0(a,b)) \wedge f_0(a,b) = b - n \times f_1(a,b) \times \big\lfloor  \frac{b}{n-1}\big\rfloor + f_1(a,b)\wedge B_k(a,f_1(a,b)) \wedge f_1(a,b) < 2$, where $f_0(a,b)$ and $f_1(a,b)$ are the first and second components of $f(a,b)$. 
Let $f_2$ be the witnessing function for (\ref{8}), i.e., $\ \Rightarrow\ f_2(a) < 2\wedge  B_k(a, f_2(a))$. Since this witness is unique, it is easy to see that $f_2(a) = f_1(a,x)$, for every $x$. 

Let $f_w$ be the function defined from $f_2$ by $\ell$-$n$ODE: 
\begin{align*}
    \frac{\partial f_w(x)}{\partial \ell} 
    &=\ \ - n \times f_2(a) \times \Big\lfloor  \frac{g(x)}{n-1}\Big\rfloor + f_2(a)
\end{align*}
with IVP as $f_w(0) =t_0$. 
Clearly, $b<p, \, B (0, t_0)\ \Rightarrow\  B (t,f_w(t))$. 
\end{proof}

The converse direction, i.e. the completeness, is more convoluted. 
We first show that an additional rule is admissible in $\tacp$. 
We will make the same abuse of notation as in Remark \ref{abuse}. 

\begin{definition}
$n$-PIND$^*(B)$ is the following inference rule: 
\begin{prooftree}
\AxiomC{(\ref{1})}
\AxiomC{(\ref{21})}
\AxiomC{(\ref{22})}
\TrinaryInfC{(\ref{3})}
\end{prooftree}
where
\begin{align} 
\Gamma,\, B(\big\lfloor\frac{a}{2}\big\rfloor, b_0,  b_1)\ &\Rightarrow\ \Delta,\, B (a, h_0(a,b_0), h_1(a,b_0,b_1))\label{1} \\
\Gamma\ &\Rightarrow\ \Delta, \, \exists d<2 \,  B_{k_0}(a, d) \label{21}\\
    \Gamma\ &\Rightarrow\ \Delta, \, \exists d<t \,  B_{k_1}(a, d) \label{22}\\
    \Gamma, \, b_0<n,\, B(0,b_0,b_1)\ &\Rightarrow\ \Delta,\, \exists y_0y_1\ B (t,y_0, y_1) \label{3}
\end{align}
and $h_0, h_1$ are defined as follows: 
\begin{itemize}
    \item $h_0(a, b_0) = b_0 - n \times k_0(a) \times \big\lfloor \frac{b_0}{n-1}\big\rfloor + k_0(a)$ where $k_0$ is esb-definable in $\tacp$. 
    \item $h_1(a,b_0, b_1) = pad(b_1, t) + k_1(a, b_0)$, where $t$ is a term, and $k_1 \leq t$ is an esb-definable function in $\mathsf{TAC^0}$[$p$].
\end{itemize}

Note that the terms $a, b_0,  b_1$ must satisfy the eigenvariable condition in (\ref{1}), (\ref{21}) and (\ref{22}). 
\end{definition}

\begin{proposition} \label{p-LIND*}
Let $B(x,y,z)$ be an esb-formula. Then, $n$-PIND$^*(B)$ is derivable in $\tacp$. 
\end{proposition}

\begin{proof}
Let $f_0$ and $f_1$ be the two functions defined by $B$, i.e., 
$$B(x, y, z) \iff y = f_0(x)\wedge z = f_1(x,y)$$ 
Let $B'(x,y) :\equiv \exists z\, B(x,y,z)$ and $B''(x,y) :\equiv \exists y\, B(x,y,z)$. 

Assume (\ref{1}), (\ref{21}), (\ref{22}) and $\tacp \vdash B(0,b_0,b_1)$ where $b_0< n$, we want to show that $\tacp \vdash \forall x \, \exists  yz\, B(x,y,z)$. Note that $\tacp$ proves $\exists !z\, B(a,b,z)$ by esb-LIND because the formula stating the uniqueness is esb: 
$$B(a,b,z_0) \wedge B(a,b,z_1) \to z_0 = z_1$$
So, the formula $B'(x,y)$ is an esb-formula. 
Clearly, $f_0$ is the function defined by $B'$. Given the condition (\ref{1}) and (\ref{21}), we obtain the premise of $n$-PIND$(B')$. Let $b$ be $b_0$ in (\ref{9}), we obtain that $f_0$ is esb-definable in $\tacp$. 

Note that (\ref{1}) also states the following: 
$$\frac{\partial f_1(x)}{\partial \ell} = \big(2^{\ell(t)}-1\big) \times f_1(x) + k_1\big(x,f_0\big(\big\lfloor\frac{x}{2}\big\rfloor\big)\big)$$
By (\ref{22}), we obtain that $k_1<t$. 
Thus, $f_1$ is definable from $b_1, k_1$ and $t$ by $\ODEup$. By Proposition \ref{prop:ODEup}, $\ODEup$ is in $\mathbb{ACDL}$; by Lemma \ref{lemma:TACFAC}, $f_1$ is esb-definable in $\tac(f_0)$. So $\tac(f_0, f_1)$ is a conservative extension of $\tac(f_0)$; therefore, $\tacp(f_1)$ is a conservative extension of $\tacp$. It is easy to show that $\tacp(f_1)\vdash  B''(a, f_1(a))$. Hence, $\tacp \vdash  \exists y\,  B''(a, y)$ by Proposition \ref{conservext}, as desired. 
\end{proof}

In order to prove the completeness, we first define how we encode the history of the function $f$ computing modulo $p$. 
We first present the notions of sequences that are available in $\tac$. For more formal detail, see Appendix 1 of \cite{CloteTakeuti}. To allege notation, we will sometimes use also $|x|$ as a synonym of $\ell(x)$.

\begin{definition}
A term $t$ is said to be a sequence definable in $\tac$ if $t$ is a pair $(a,b)$ s.t. 
\begin{itemize}
    \item $a = \bt{pad}(t_0, \bt{pad}(t_1, \ldots, \bt{pad}(t_{m-1}, t_m) \ldots))$
    \item for every $i$, $|t_i| = |b|+1$
\end{itemize}
where the length $\bt{lh}(t)=m+1$ of $t$ is computable from $a$ as $\mu x< |a|, \ |a| = |x| \times (|b|+1)$. 
\end{definition}

\begin{remark}
With abuse of notation, we use $\langle t_0, \ldots, t_m\rangle$ to denote the $t$ above. 
\end{remark}


\begin{proposition} \label{seq}
Let $w$ and $u$ be sequences definable in $\tac$, then the following predicates can be esb-defined: \begin{itemize}
    \item $\bt{seq}(w)$: $w$ is a sequence. 
    \item $w(i)$: the $i$-th element of $w$. 
    \item $w[i]$: $\langle w(0), \ldots, w(i) \rangle$
    \item $w*u$: concatenation of $w$ and $u$. 
\end{itemize}
\end{proposition}

For each $x$, we can construct a history of dividing $x$ by $2$, as follows: 
$$w_x\ \ :=\ \ \langle\ 0, 1, \ldots, \big\lfloor\frac{x}{2}\big\rfloor, x \ \rangle$$
Then, the sequence below represents the history of computing $f(x)$: 
$$f[w_x]\ \ :=\ \ \langle\ f(0), f(1) \ldots, f(\big\lfloor\frac{x}{2}\big\rfloor), f(x)\ \rangle$$



However, $w_x$ cannot be directly encoded in $\tacp$, because every element of $w_x$ is of increasing size and $\tac$ can only handle sequences where every element is of fixed size. Hence, we introduce an encoding of the history of computation. It is possible to reconstruct $w_x$ from the following sequence in $\tac$: 
\begin{align*}
\hat{w}_x\ \  & :=\ \  \langle\  x - 2\times \big\lfloor\frac{x}{2}\big\rfloor, \big\lfloor\frac{x}{2}\rfloor - 2\times \big\lfloor\frac{x}{4}\rfloor, \ldots\ldots, \big\lfloor\frac{x}{2^{|x|-1}}\rfloor - 2\times \big\lfloor\frac{x}{2^{|x|}}\rfloor\ \rangle\\
   & =\ \  \langle \fun{BIT}(x,1),,\ldots, \fun{BIT}(x,\ell(x)-1),  \fun{BIT}(x,\ell(x)) \ \rangle
\end{align*}
whose elements are in $\{0,1\}$. 
Recall that $\big\lfloor\frac{x}{2^{|y|}}\big\rfloor=\mathsf{MSP}(x,|y|)$ and that $\fun{BIT}$ is axiomatized from $\bt{MSP}$ in $\tac$. 

\begin{lemma} \label{wx}
The following maps can be computed in $\fac$: 
$$W\  :\  x \  \mapsto \  \hat{w}_x, \quad 
    X\  :\  \hat{w}_x \ \mapsto \  x$$
\end{lemma}

\begin{proof}
$W$ can be defined by CRN. $X$ can be constructed by operators manipulating sequences in Proposition \ref{seq}. 
\end{proof}

In the following, let $f$ be defined by $\ell$-$n$ODE from $b$ and $k$ and $k \in \ACCp{n}$. Since $f$ is also a function in $\ACCp{n}$, we will show that $f$ is esb-definable in $\tacp$. 

To this end, we must first obtain a correct way of representing $B_f$. 
We first define two auxiliary formulas, based on the notion of sequence that is available in $\tac$.

\begin{definition}
$B_1(w,x) :\equiv\ \bt{Seq}(w) \, \wedge\, \bt{lh}(w) = |x| \,
     \wedge \, \forall i < |x|, \ w(i) = \fun{BIT}(x, i+1)$. 
     
    $B_2(w, u) :\equiv \ \bt{seq}(w) \wedge \bt{Seq}(u) \wedge \bt{lh}(w) = \bt{lh}(u) \wedge u(0) = b\ \wedge 
     \forall i < \bt{lh}(u)-1, \ u(i+1) = u(i)- n \times k(X(w[i])) \times \lfloor \frac{u(i)}{n-1}\rfloor + k(X(w[i]))$. 
\end{definition}

Intuitively, in $B_1$, $w$ does the job of $\hat{w}_x$; in $B_2$, $u$ constructs $f[w_x]$ from $w$. 


\begin{example}
Let $x=11$, then
\begin{align*}
    w_{11} &= \langle \ 0,1,2,5,11\ \rangle \mbox{ and } \hat{w}_{11} =\langle \ 1, 1, 0,1\ \rangle
\end{align*}


Let $k(1)= k(11) = 1$, $k(2)= k(5) = 0$, $p=2$ and the initial value $b=1$. Then, $w = w_{11}$, $f[w_{11}] = u *\langle y\rangle$ where 
\begin{align*}
    u &= \langle\  1,0,0,0 \ \rangle \mbox{ and } y=1 
\end{align*}
\end{example}

Now, we define the formula $B_f$ for which we shall prove the completeness. 

\begin{definition} \label{convention}
$B_f(x, y, u)  :\equiv\ \exists w\ B_1(w, x) \wedge B_2(w,u)\wedge y=u(|x|) - k(x) \times n \times \lfloor \frac{u(|x|)}{n-1}\rfloor + k(x)$. 

We adapt the convention that if $x= 0$, then $y=b$ and $u = \langle \rangle$, and $u(|x|)$ has a unique value $v\in [0, \ldots, n-1]$ s.t. $b = v - k(0)\times n \times \lfloor \frac{v}{n-1}\rfloor + k(0)$.  
\end{definition}

\begin{remark}
$b\in [0,\ldots, n-1]$ is the value of $f(0)$ and $v$ does not correspond to any output of $f$. However, $v$ is unique and must exist, as we can define it as follows: 
$$v = \begin{cases}
    b &\text{if }k(0) = 0 \\
    b-1 &\text{if } k(0) = 1\ \&\ b \neq 0\\
    n-1 &\text{otherwise}
\end{cases}$$
\end{remark}

It is from the meaning of $B_f$ that it defines the function $f$ and its history of computation. 

\begin{proposition} \label{phif}
In $\mathbb{N}$, 
$$B_f(x, y, u) \ \ \iff\ \ y = f(x)\ \wedge\ u = w_{f(\lfloor\frac{x}{2}\rfloor)}$$
\end{proposition}

It follows by definition of $B_f$ that for every $x$, the values of $y$ and $u$ are unique (in $\mathbb{N}$). The following proposition establishes this uniqueness in $\tac$, using the same technique as the proof of uniqueness in Proposition \ref{p-LIND*}. 

\begin{proposition}
For any $a$, 
$\mathsf{TAC^0}\, \vdash\, \exists ! w\ B_1(w,a)$
\end{proposition}

\begin{corollary} \label{esb}
$B_f$ is an esb-formula in $\tac$ (hence also in $\tacp$). 
\end{corollary}

\begin{lemma} \label{precompleteness}
Assume that $k$ is esb-definable in $\tacp$. Then, $f$ is esb-definable in $\tacp$. 
\end{lemma}
\begin{proof}
We will show that $B_f$ is provably total in $\tacp$, which means that $f$ is esb-definable in $\tacp$ by Proposition \ref{phif}. We need to show the following two statements: 
\begin{itemize}
    \item[(i)] $\tacp\vdash \forall x \exists yu\, B_f(x,y,u)$
    \item[(ii)] $\tacp\vdash \forall x \forall yu \forall y'u'\, B_f(x,y,u)\wedge B_f(x,y',u')\to y=y' \wedge u = u'$
\end{itemize}

By Corollary \ref{esb}, $B_f$ is an esb-formula. So, 
we need to show that (\ref{1}), (\ref{21}) and (\ref{22}) of p-LIND$^*(B_f)$ are provable in $\tacp$ in order to show (i). (\ref{1}) consists of the following: 
\begin{align*} 
    \mathsf{TAC^0}[p]\ \vdash\ B_f(\lfloor \frac{a}{2}\rfloor, d, u_0) \to B_f(a, d - k(a) \times n \times \lfloor \frac{d}{n-1}\rfloor + k(a), u_1)
\end{align*}
for some $u_1$ satisfying $B_f$. Now we construct $u_1$ in $\facp$. 
\vspace{0.2cm}

Let $n = |a|$ and $w = W(a)$. 
Let $u_0$ be the sequence satisfying $u_0(0) = b\ \wedge\ \bt{lh}(u_0) = n-1\ \wedge\ \forall i < n-1\ u_0(i+1) = u_0(i)- n \times k(X(w[i])) \times \lfloor \frac{u_0(i)}{n-1}\rfloor + k(X(w[i]))\ \wedge\  d = u_0(n-2) - p\times k(\lfloor\frac{a}{2}\rfloor) \times \lfloor\frac{u_0(n-2)}{n-1}\rfloor + k(\lfloor\frac{a}{2}\rfloor)$. 
Then, $u_1$ can be defined as follows: 
$$u_1\ \ =\ \ \bt{pad}(u_0, p)+  d $$

Since $a$ is arbitrary and $k$ is esb-definable by assumption, we get $\tacp \vdash$ (\ref{21}). 
By Proposition \ref{phif}, (\ref{22}) is $\exists d< p\, A(a, d)$ where $A(a,b)\iff f(\lfloor\frac{a}{2}\rfloor) = b$. Hence, we also have $\tacp \vdash$ (\ref{22}). 
With $\tacp \vdash B_f(0, b, \langle \rangle)$, we get (i) by $n$-PIND$^*(B_f)$. 
\vspace{0.2cm}

Now we show (ii). Note that the construction of $u_1$ from $u_0$ shows that if there are unique $d_0$ and $u_0$ satisfying $B_f(\lfloor \frac{a}{2}\rfloor, d_0, u_0)$, then there are unique $d_1$ and $u_1$ s.t. $B_f(a, d_1, u_1)$. Namely, $d_1 = d_0 - k(a) \times n \times \lfloor \frac{d_0}{n-1}\rfloor + k(a)$. This reasoning takes place in $\tacp$. Since the formula stating the uniqueness $B_f(x,y,u)\wedge B_f(x,y',u')\to y=y' \wedge u = u'$ is an esb-formula, (ii) follows from the esb-LIND. 
\end{proof}

Any function in $\mathbb{ACDL}$ is esb-definable in $\tac$; thus, Lemma \ref{precompleteness} shows the completeness. Together with Lemma \ref{witness}, we obtain Theorem \ref{completeness}.


\end{toappendix}

\section{Conclusion and perspective}\label{sec:conclusion}

The paper introduces \emph{new} and \emph{uniform} machine-independent characterizations for constant-depth small circuit classes including modulo counting for any constant, in both the recursion- and proof-theoretic settings.
The key novelty of our approach lies in its reliance on very simple (i.e.,~strict) recursion schemas defined by discrete ODEs. 
This apparently simple shift in perspective, initiated in~\cite{BournezDurand} for $\FP$ and scalable from the discrete to the continuous realm, has already proven to be highly effective in capturing (small circuit) classes, possibly including (unrestricted) counting gates~\cite{ADK24a,ADK25}. 
In the ODE framework it becomes natural to ``parametrize'' the recursion schema in ways that are not natural (hence were not considered) in the classical recursion-theoretic context.
This is precisely what allows us to obtain the first implicit characterizations for constant-depth circuit computation that extend to modulo counting. 
%
%
By simply equipping algebras with properly-defined (still strict) schemas, the correspondence  between strict schemas and $\AC^0$ (see also~\cite{ADK24a,ADK25}) extends to $\ACCp{n}$ classes.
This was not the case for characterizations presented so far in the literature; as mentioned, $\ACCp{2}$ and $\ACCp{6}$ were captured by restricting $k$-BRN, the schema defining $\NC^1$ (for $k\ge 4$), in a way that did not allow for a generalization to other modulo counting classes.
%
Furthermore, by mirroring the behavior of the new  schemas, we present bounded arithmetic theories capturing the corresponding circuit classes.
%
As an additional positive feature, our axioms are conceptually simpler than the corresponding rules by~\cite{CloteTakeuti} and the naturalness of the approach leads to completely self-contained and technically straightforward proofs, which were not the case for (partial) characterizations obtained in the less flexible frameworks of~\cite{CloteTakeuti,ADK25}.
%
%
%

Overall, this work is conceived as a first step towards a more general study of (circuit) complexity from the viewpoint of ODEs and several research directions are still open.
First, it would be  profitable to better understand connections between strict and non-strict schemas (the latter, allowing one to call more freely for preceding values in their definition), to delineate the frontier between  computation in constant and (bounded) logarithmic-depth in this machine-independent framework. 
In this vein, the role of the change of variable as a method of reduction should be investigated much further. 
Other natural and intriguing questions concern $\ell$-ODE schemas allowing for exact division or over the reals.
In the longer term, it would be interesting to investigate separation between classes or bounds, possibly in terms of syntactical constraints over ODE schemas (or bounded theories) and relying on tools from the difference calculus.
Regarding bounded arithmetic, this work is the first attempt to link ODE-based characterizations and proof-theoretical approaches. Although here we can rely on multiple simplifications due to the limited small circuit settings (i.e.,~the fact that we are mostly dealing with ``small sequences'', whose encoding is especially simple), our results seem to point to promising directions to ``translate'' ODE-based algebras into corresponding rule systems (and vice versa).
At first, it would  be natural to see whether this approach can be used to define ODE-designed deduction rules for  $\mathbb{LDL}$, the algebra defined in~\cite{BournezDurand} for $\FP$.
Follow-up, more general studies might help to  develop an original machine-independent approach to relate recursion- and proof-theoretical views of complexity.

\arxiv{(considering interesting separation results in the vein of Cor.~\ref{cor:nODE}).}
%
%
%

\bibliographystyle{plain}
\bibliography{reference}

\appendix

\end{document}

%% file: macros.tex
\newcommand{\TC}{\mathbf{FTC}}

\newcommand{\cc}[1]{\mathbf{#1}}

\newcommand{\bth}[1]{\mathsf{#1}}
\newcommand{\bt}[1]{\textsc{#1}}

\newcommand{\FP}{\mathbf{FP}}
\newcommand{\NC}{\mathbf{FNC}}
\newcommand{\AC}{\mathbf{FAC}}
\newcommand{\ACCt}{\mathbf{FAC}^0\mathbf{[2]}}
\newcommand{\ACCp}[1]{\mathbf{FAC}^0\mathbf{[{#1}]}}

\newcommand{\ACDL}{\mathbb{ACDL}}

\newcommand{\tu}[1]{\mathbf{#1}}
 
\newcommand{\longv}[1]{}

\newcommand{\fun}[1]{\mathsf{#1}}

\newcommand{\Nat}{\mathbb{N}}

\newcommand\N{\mathbb{N}}

\newcommand{\ODEup}{\ell\text{-ODE}_2^\uparrow}

 \usepackage{bussproofs}

 \newcommand{\raus}[1]{}